\documentclass[pra,aps,superscriptaddress,twocolumn,nofootinbib,longbibliography,accepted=2024-08-28]{quantumarticle}
\pdfoutput=1

\usepackage{graphicx,color,amsfonts,amsmath,amssymb,amsthm,appendix,bbold,bbm,bm,braket,centernot,comment,enumerate,enumitem,float,hyperref,mathdots,mathtools,pgfplots,soul,subfigure,thmtools}

\usepackage[numbers]{natbib}
\usepackage[T1]{fontenc}

\hypersetup{colorlinks=true,linkcolor=blue,citecolor=blue,urlcolor=blue}

\pgfplotsset{compat=1.18} 

\theoremstyle{definition}
\newtheorem{definition}{Definition}

\newtheorem{theorem}{Theorem}
\newtheorem{prop}{Proposition}
\newtheorem{corollary}{Corollary}

\def\H{ {\cal H} }

\def\>{\rangle}
\def\<{\langle}

\newcommand{\expval}[1]{\left\langle #1\right\rangle}

\newcommand{\ketbra}[2]{\ensuremath{\left| #1 \vphantom{#2} \right\rangle\!\!\left\langle #2 \vphantom{#1} \right|}}
\newcommand{\tr}[1]{\mathrm{Tr}\left( #1 \right)}

\newcommand{\iden}{\mathbbm{1}}

\newcommand{\Uuparrow}{\mathrel{\rotatebox[origin=c]{90}{$\Rrightarrow$}}}
\newcommand{\Ddownarrow}{\mathrel{\rotatebox[origin=c]{270}{$\Rrightarrow$}}}
\newcommand{\ssUuparrow}{\mathrel{\scalebox{.5}{\rotatebox[origin=c]{90}{$\Rrightarrow$}}}}
\newcommand{\sUuparrow}{\mathrel{\scalebox{.75}{\rotatebox[origin=c]{90}{$\Rrightarrow$}}}}

\newcommand{\ssDdownarrow}{\mathrel{\scalebox{.5}{\rotatebox[origin=c]{270}{$\Rrightarrow$}}}}
\newcommand{\sDdownarrow}{\mathrel{\scalebox{.75}{\rotatebox[origin=c]{270}{$\Rrightarrow$}}}}

\renewcommand{\v}[1]{\ensuremath{\boldsymbol #1}}

\newcommand{\twoheaduparrow}{\mathrel{\rotatebox[origin=c]{-90}{$\twoheadleftarrow$}}}

\newcommand{\twoheaddownarrow}{\mathrel{\rotatebox[origin=c]{90}{$\twoheadleftarrow$}}}

\usepackage{titlesec}

\begin{document}

\title{Perfect quantum protractors}

\author{Micha\l{} Piotrak}
\affiliation{Department of Physics, Imperial College London, London SW7 2AZ, United Kingdom}
\affiliation{Department of Physics and Astronomy, University College London, London WC1E 6BT, United Kingdom}

\author{Marek Kopciuch}
\affiliation{Doctoral School of Exact and Natural Sciences, Jagiellonian University, Faculty of Physics, Astronomy and Applied Computer Sciences, \\30-348 Kraków, Poland}
\affiliation{Faculty of Physics, Astronomy and Applied Computer Science, Jagiellonian University, 30-348 Kraków, Poland}

\author{Arash Dezhang Fard}
\affiliation{Doctoral School of Exact and Natural Sciences, Jagiellonian University, Faculty of Physics, Astronomy and Applied Computer Sciences, \\30-348 Kraków, Poland}
\affiliation{Faculty of Physics, Astronomy and Applied Computer Science, Jagiellonian University, 30-348 Kraków, Poland}

\author{Magdalena Smolis}
\affiliation{Faculty of Physics, Astronomy and Applied Computer Science, Jagiellonian University, 30-348 Kraków, Poland}

\author{Szymon Pustelny}
\affiliation{Faculty of Physics, Astronomy and Applied Computer Science, Jagiellonian University, 30-348 Kraków, Poland}

\author{Kamil Korzekwa}
\affiliation{Faculty of Physics, Astronomy and Applied Computer Science, Jagiellonian University, 30-348 Kraków, Poland}

\begin{abstract}

In this paper we introduce and investigate the concept of a \emph{perfect quantum protractor}, a pure quantum state $\ket{\psi}\in\H$ that generates three different orthogonal bases of $\H$ under rotations around each of the three perpendicular axes. Such states can be understood as pure states of maximal uncertainty with regards to the three components of the angular momentum operator, as we prove that they maximise various entropic and variance-based measures of such uncertainty. We argue that perfect quantum protractors can only exist for systems with a well-defined total angular momentum $j$, and we prove that they do not exist for $j\in\{1/2,2,5/2\}$, but they do exist for $j\in\{1,3/2,3\}$ (with numerical evidence for their existence when $j=7/2$). We also explain that perfect quantum protractors form an optimal resource for a metrological task of estimating the angle of rotation around (or the strength of magnetic field along) one of the three perpendicular axes, when the axis is not \emph{a priori} known. Finally, we demonstrate this metrological utility by performing an experiment with warm atomic vapours of rubidium-87, where we prepare a perfect quantum protractor for a spin-1 system, let it precess around $x$, $y$ or $z$ axis, and then employ it to optimally estimate the rotation angle. 

\end{abstract}

 \maketitle


\section{Introduction}

Uncertainty relations are one of the most well-known aspects of quantum mechanics. Despite being originally formulated more than a century ago~\cite{heisenberg1927anschaulichen,robertson1929uncertainty}, they are still an active field of research~\cite{busch2014colloquium,coles2017entropic}. This counter-intuitive feature of nature that knowing one property of a quantum system with certainty necessarily means that we must be uncertain about another complementary property is usually considered as a fundamental metrological restriction. Thus, we can never measure perfectly conjugate variables, such as position and momentum, time and energy, or angular position and angular momentum. 

However, one can also look at this from a different angle and see quantum uncertainty not as a metrological restriction, but rather as a metrological resource. First, note that the operators corresponding to one of the pair of conjugate variables generate change of the other variable: momentum operator is a generator of translations in space, energy operator (Hamiltonian) is a generator of time evolution, and angular momentum operators generate rotations in space. Then, it is clear that the eigenstates of a given generator (i.e., states with a well-defined value and no uncertainty) are useless for measuring the conjugate variable. For example, energy eigenstates are known to be stationary and do not evolve in time -- as such, they are useless for measuring time. On the other hand, a quantum system prepared in a superposition of different energy eigenstates (so with some energy uncertainty) can act as a clock, since it can non-trivially evolve in time to another distinguishable state~\cite{mandelstam1991uncertainty}.

This concept was formalised by introducing the notion of \emph{quantum clocks}~\cite{peres1980measurement}, which was then generalised in Ref.~\cite{buzek1999optimal}. The basic idea is that a clock can be realised by a finite-dimensional quantum system, whose ``ticks'' at discrete times correspond to evolutions into different mutually orthogonal states. Projective measurement in the basis formed by these states allows for accurate, although discrete, measurements of time. The maximal number of different moments in time that can be then perfectly distinguished by a $d$-dimensional quantum clock is given by~$d$, and this can happen only for the clock initialised in an optimal state of uniform superposition of energy eigenstates (i.e., for the state maximising energy uncertainty).

In this paper, we generalise the original construction of quantum clocks to quantum systems being able to optimally measure more than one variable. More precisely, we introduce and investigate the notion of \emph{perfect quantum protractors} -- $d$-dimensional quantum states that can be used to perfectly distinguish $d$ rotation angles around any of the three perpendicular axes. In other words, under rotations around $x$, $y$ and $z$ axes these states generate three different orthogonal bases, and they also simultaneously maximise the uncertainty of the three components of angular momentum. We show that such optimal states can only exist for systems with a well-defined angular momentum $j$ and, surprisingly, we prove that their existence depends on $j$. We analyse in detail small-dimensional cases of $j<4$ and explain metrological setups for which perfect quantum protractors are optimal resources. Finally, we perform an experiment with atomic vapours of rubidium-87 playing the role of spin-1 perfect quantum protractors employed to optimally estimate the rotation angle.

The paper is structured as follows. First, in Sec.~\ref{sec:setting}, we set the scene by recalling the mathematical formalism of angular momentum and formally defining the central concepts of optimal and perfect quantum protractors. Next, in Sec.~\ref{sec:existence}, we present results on the existence and non-existence of these states for spin-$j$ systems. We then proceed to discussing the properties of perfect quantum protractors in Sec.~\ref{sec:properties}, focusing mainly on the fact that they are pure states that maximise uncertainty about angular momentum measurements, and describing their entanglement features. Next, Sec.~\ref{sec:metrology} is devoted to justifying the utility of perfect quantum protractors for metrological tasks of measuring rotation angles. In Sec.~\ref{sec:experiment} we present the experiment we performed with rubidium vapour, where we prepared perfect quantum protractor states of spin-1 systems and used them to optimally estimate rotation angles. Finally, we conclude the paper and provide outlook for future work in Sec.~\ref{sec:conclusions}.


\section{Setting the scene}
\label{sec:setting}


\subsection{Mathematical background}
\label{sec:math}

Rotations in three dimensions are generally described by a Lie group SO(3), i.e., a special orthogonal group of dimension $3$. Orthogonality ensures that the transformation preserves the inner product, hence it is an isometry. Speciality implies that the parity is conserved as well, hence the group is restricted to proper rotations. The group SO(3) is double-covered by SU(2), the special unitary group of dimension 2, so equivalently rotations can be described by elements of SU(2). This is very convenient when considering quantum particles with angular momentum and will be used throughout this paper.

The Lie algebra $\mathfrak{su(2)}$ generating SU(2) is defined by the following commutation relation between its three generators:
\begin{equation}
    \label{eq:lie-algebra}
    [{J}_k,{J}_l]=i\epsilon_{klm}{J}_m,
\end{equation}
where $\epsilon_{klm}$ is the Levi-Civita tensor. Since SU(2) is a 3-parameter Lie group, each of its elements can be formed from the elements of $\mathfrak{su(2)}$ via exponentiation:
\begin{equation}
    \forall \ {R(\pmb{\theta})\in \text{SU(2)}}:\quad R(\pmb{\theta})=e^{-i\pmb{\theta}\cdot\pmb{{J}}},
\end{equation}
where $\pmb{\theta}=(\theta_x,\theta_y,\theta_z)$ is a vector of real numbers and $\pmb{{J}}=({J}_x,{J}_y,{J}_z)$. To denote rotations around $x$, $y$ and~$z$ axes we will use a shorthand notation of $R_x$, $R_y$ and $R_z$.

Matrix representations of $\mathfrak{su(2)}$ of size $d$ are associated with angular momentum operators along different axes: ${J}_x$, ${J}_y$ and ${J}_z$ correspond to angular momenta along three perpendicular axes, and so $\hat{\v{n}}\cdot \pmb{{J}}$, with $\hat{\v{n}}$ denoting a unit vector, represents angular momentum along axis $\hat{\v{n}}$\footnote{Throughout the paper we will slightly abuse the notation with $R(\v{\theta})$ and ${J}_k$ denoting both elements of the Lie group and algebra, as well as their matrix representations of size $d$.}. In particular, irreducible representations of size $(2j+1)$ correspond to spin-$j$ systems, whose Hilbert space we will denote by $\H_j$, and for which one has
\begin{equation}
    J^2:=J_x^2+J_y^2+J_z^2=j(j+1).    
\end{equation}
In the irreducible case, the eigenstates of $\hat{\v{n}}\cdot \pmb{{J}}$ form a complete and orthonormal basis of $\H_j$, so any state with a total angular momentum $j$ can be expressed as:
\begin{equation}
    \label{eq:irrep_decomp}
    \forall \ {\ket{\psi} \in \H_j}:\quad \ket{\psi}=\displaystyle \sum_{m=-j}^{j}c_{\hat{\v{n}},m}\ket{j,m}_{\hat{\v{n}}}, 
\end{equation}
where $\ket{j,m}_{\hat{\v{n}}}$ is the eigenstates of $\hat{\v{n}}\cdot \pmb{{J}}$ with eigenvalue~$m$. 

Now, a state $\ket{\psi}$ of a quantum system transforms under the rotation described by $R(\pmb{\theta})$ as follows:
\begin{equation}
     \ket{\psi}\xrightarrow{R(\pmb{\theta})} e^{-i{{\pmb{\theta}\cdot\v{J}}}}\ket{\psi}=e^{-i\theta\hat{\pmb{n}}\cdot{\pmb{J}}}\ket{\psi},
\end{equation}
where we decomposed $\pmb{\theta}$ into a unit vector times its magnitude, $\pmb{\theta}=\theta \hat{\pmb{n}}$. In particular, for spin-$j$ systems, we can use the decomposition from Eq.~\eqref{eq:irrep_decomp} to see that the rotation by angle $\theta$ around axis $\hat{\pmb{n}}$ transforms the state of the system as follows:
\begin{align}
    \label{eq:rot}
   \!\!\! \ket{\psi}=\!\!\!\sum_{m=-j}^{j}\! c_{\hat{\v{n}},m}\ket{j,m}_{\hat{\pmb{n}}}\xrightarrow{R(\theta\hat{\v{n}})}\!\! \sum_{m=-j}^{j}\! c_{\hat{\v{n}},m}e^{-i m \theta}\ket{j,m}_{\hat{\pmb{n}}}\!.\!
\end{align}
For a general quantum system, the corresponding Hilbert space $\H$ can be uniquely decomposed into a direct sum of Hilbert spaces corresponding to irreducible representations of SU(2) as follows:
\begin{equation}
    \label{eq:decomposition}
    \H=\bigoplus_{j}\H_j^{\otimes g_j},
\end{equation}
with $j$ taking integer or half-integer values and $g_j$ denoting the degeneracy of a subspaces $\H_j$. For example, a system consisting of four spin-$1/2$ particles can be decomposed as:
\begin{align}
    (\H_{1/2})^{\otimes 4}&=\H_0^{\otimes 2}\oplus\H_1^{\otimes 3}\oplus \H_2.
\end{align}
Thus, an arbitrary quantum state $\ket{\Psi}$ living in a general Hilbert space $\H$ can be expressed as:
\begin{equation}
    \label{general}
    \ket{\Psi}=\sum_{j=0}^{j_{\mathrm{max}}} \sum_{m=-j}^j \sum_{\alpha=1}^{g_j} c_{\hat{\v{n}},j,m,\alpha} \ket{j,m,\alpha}_{\hat{\v{n}}},
\end{equation}
where $j_{\mathrm{max}}$ is the total angular momentum of the highest irrep appearing in the decomposition of $\H$.


\subsection{Formal statement of the problem}
\label{sec:formal}

In this work we investigate special families of quantum states that we call \emph{quantum protractors}.

\begin{definition}[Quantum protractors]
    A quantum protractor of order $n$ with respect to axis $\hat{\v{n}}$ is a quantum state $\ket{\psi}\in\H$ with $\dim \H =d$ for which there exist angles $\{\theta_1,\dots,~\theta_n\}$ such that
    \begin{equation}
        \forall~1\leq k < l\leq n:\quad \bra{\psi} R^\dagger(\theta_{k}\hat{\v{n}})R(\theta_{l}\hat{\v{n}}) \ket{\psi}=0,
    \end{equation}
    i.e., $\ket{\psi}$ rotated around $\hat{\v{n}}$ by these angles forms a set of $n$ mutually orthogonal states. A quantum protractor of order $d$ is called \emph{optimal}. A state that is an optimal quantum protractor simultaneously with respect to $r$ perpendicular axes is said to be of rank $r$, and a quantum protractor with the maximal rank equal to 3 is called \emph{perfect}.

\end{definition}

Our main goal is to solve the existence problem of perfect quantum protractors in Hilbert spaces of various dimensions, capture their properties and understand the consequences of their (non-)existence.


\section{Existence}
\label{sec:existence}

We first observe that in our search for optimal quantum protractors, we only need to consider states living in Hilbert spaces $\mathcal{H}_j$, i.e., states of quantum systems with well-defined total angular momentum $j$. To see this, first consider an arbitrary state $\ket{\Psi}\in\H$ of a general system given by Eq.~\eqref{general}. The overlap between $\ket{\Psi}$ rotated by angle $\theta_k$ about axis $\hat{\v{n}}$ and $\ket{\Psi}$ rotated by angle $\theta_l=\theta_k-\theta$ about the same axis is given by:
\begin{align}
    \!\!\!\bra{\Psi} e^{-i\theta\hat{\pmb{n}}\cdot{\pmb{J}}}\ket{\Psi}&=\displaystyle\sum_{j=0}^{j_{\mathrm{max}}}\sum_{m=-j}^j \sum_{\alpha=1}^{g_j} |c_{\hat{\v{n}},j,m,\alpha}|^2e^{-im\theta}\nonumber\\
    \!\!\!&=\sum_{m=-j_{\mathrm{max}}}^{j_{\mathrm{max}}}\!\!\! e^{-im\theta}\sum_{j=|m|}^{j_{\mathrm{max}}}\sum_{\alpha=1}^{g_j}|c_{\hat{\v{n}},j,m,\alpha}|^2.
\end{align}
Introducing
\begin{equation}
    p_m:=\sum_{j=|m|}^{j_{\mathrm{max}}}\sum_{\alpha=1}^{g_j}|c_{\hat{\v{n}},j,m,\alpha}|^2,
\end{equation}
so that $\v{p}$ is a normalised probability distribution, we get
\begin{equation}
        \bra{\Psi} e^{-i\theta\hat{\pmb{n}}\cdot{\pmb{J}}}\ket{\Psi}=\sum_{m=-j_{\mathrm{max}}}^{j_{\mathrm{max}}}e^{-im\theta} p_m.
\end{equation}
Exactly the same overlaps would be obtained if instead we rotated a state $\ket{\psi}\in \H_{j_{\mathrm{max}}}$ from Eq.~\eqref{eq:irrep_decomp} with \mbox{$c_{\hat{\v{n}},m}=\sqrt{p_m}$}. This proves that the number of orthogonal states that can be obtained via rotations of a system with a composite Hilbert space $\H$ is bounded by the dimension of the highest irrep $\H_{j_{\mathrm{max}}}$ appearing in the decomposition from Eq.~\eqref{eq:decomposition}. Since for an optimal quantum protractor we require this number to be equal to the dimension of the system, without loss of generality we can limit our considerations to states living in $\H_j$.


\subsection{Optimal quantum protractors of rank 1}
\label{sec:rank1}

We start with a simple observation that characterises all optimal quantum protractors of rank 1 as phase states~\cite{barnett1992limiting}.

\begin{prop}
\label{prop:optimal1}
A state $\ket{\psi}\in\H_j$ is an optimal quantum protractor with respect to axis $\hat{\v{n}}$ if and only if it is of the form
\begin{equation}
    \label{eq:optimal1}
    \ket{\psi}=\frac{1}{\sqrt{2j+1}} \sum_{m=-j}^j e^{i\phi_m} \ket{j,m}_{\hat{\v{n}}}.
\end{equation}
\end{prop}
\begin{proof}
    We first prove by direct calculation that $\ket{\psi}$ from Eq.~\eqref{eq:optimal1} is an optimal quantum protractor. We start by choosing $(2j+1)$ angles $\theta_k=2\pi k/(2j+1)$ for \mbox{$k\in\{0,\dots,2j\}$}. Then, the overlap between $\ket{\psi}$ rotated by angle $\theta_k$ about axis $\hat{\v{n}}$ with $\ket{\psi}$ rotated by angle $\theta_l$ about the same axis is given by:
    \begin{align}
        \bra{\psi} e^{i(\theta_k-\theta_l)\hat{\pmb{n}}\cdot{\pmb{J}}}\ket{\psi}=\frac{1}{2j+1}\sum_{m=-j}^{j}e^{i\frac{2\pi (k-l)m}{2j+1}}=\delta_{kl}.
    \end{align}
    Since there exist $\dim \H_j = (2j+1)$ rotation angles that transform $\ket{\psi}$ into a set of mutually orthogonal states, $\ket{\psi}$ is an optimal quantum protractor.
    
    We will now prove that every optimal quantum protractor must be of the form given in Eq.~\eqref{eq:optimal1}. For $\ket{\psi}$ to be an optimal quantum protractor, its $(2j+1)$ rotated versions must be mutually orthogonal. However, according to Eq.~\eqref{eq:rot}, rotations only introduce phase factors in the relevant angular momentum eigenbasis, and do not modify probabilities. Thus, we require the existence of $(2j+1)$ orthogonal states $\ket{\psi_k}$ with $k\in\{0,\dots 2j\}$ such that 
    \begin{equation}
        \ket{\psi_k}=\sum_{m=-j}^j \sqrt{p_m} e^{i\phi_{k,m}} \ket{j,m}_{\hat{\v{n}}},
    \end{equation}
    with $p_m=|\bra{\psi} j,m\rangle_{\hat{\v{n}}}|^2$. However, Proposition~4 from Ref.~\cite{korzekwa2019distinguishing} clearly states that this is possible only if $\max_m p_m \leq 1/(2j+1)$, which enforces all $p_m$ to be equal to $1/(2j+1)$, which in turn enforces $\ket{\psi}$ to be of the form from Eq.~\eqref{eq:optimal1}.
\end{proof}


\subsection{Optimal quantum protractors of rank 2}
\label{sec:rank2}

We now proceed to proving the existence of optimal quantum protractors of rank 2 in Hilbert spaces $\H_j$ for all $j$. From Proposition~\ref{prop:optimal1}, we know that in order for \mbox{$\ket{\psi}\in \H_j$} to be an optimal quantum protractor with respect to axes $\hat{\v{n}}$ and $\hat{\v{m}}$, it has to satisfy the following:
\begin{align}
    \ket{\psi}&=\frac{1}{\sqrt{2j+1}} \sum_{m=-j}^j e^{i\phi_m} \ket{j,m}_{\hat{\v{n}}} \nonumber\\
    &=  \frac{1}{\sqrt{2j+1}} \sum_{m=-j}^j e^{i\varphi_m} \ket{j,m}_{\hat{\v{m}}}.
\end{align}
In other words, we are looking for a state that is a uniform superposition of all basis states for two different bases, the eigenbases of $\hat{\pmb{n}}\cdot{\pmb{J}}$ and $\hat{\pmb{m}}\cdot{\pmb{J}}$. Luckily for us, this problem was investigated in the past in a different context, and so we can employ the following result from Refs.~\cite{korzekwa2014operational,idel2015sinkhorn}.
\begin{theorem}[Theorem~1 of Ref.~\cite{korzekwa2014operational}]
For any two bases $\{ \ket{a_n} \}$ and $\{ \ket{b_n} \}$ of a $d$-dimensional Hilbert space there exist at least $2^{d-1}$ states $\ket{\psi}$ that are unbiased in both bases, i.e.:
\begin{equation}
    \forall n:\quad |\braket{a_n|\psi}|^2=|\braket{b_n|\psi}|^2=\frac{1}{d}.
\end{equation}
\end{theorem}

The above then straightforwardly leads to the following corollary:
\begin{corollary}
    \label{cor:rank2}
    In every Hilbert space $\H_j$, there exist at least $4^{j}$ optimal quantum protractors with respect to any two axes and so, in particular, there exist optimal quantum protractors of rank 2.
\end{corollary}
We will provide examples of optimal quantum protractors of rank 2 in the next section while discussing the existence of perfect quantum protractors.


\subsection{Perfect quantum protractors}
\label{sec:analysis}

For a quantum protractor to have rank $3$ it must satisfy the condition specified in Proposition~\ref{prop:optimal1} for each of the three perpendicular axes. Thus, in what follows, we will start with the most general states belonging to $\H_j$ and enforce them to yield uniform distributions in each of the three angular momentum eigenbases. To simplify the  notation, we will use $x$, $y$ and $z$ subscripts to denote unit vectors along three Cartesian axes, and we will skip $j$ in $\ket{j,m}$ as it will be fixed in each of the following sections. Also, for $m$, instead of using $\pm 1/2$, $\pm 1$, etc., we will use a tidier and self-explanatory notation employing the following symbols: \mbox{$\{ \Uuparrow,\Uparrow,\uparrow,0,\downarrow,\Downarrow,\Ddownarrow\}$}. Finally, for the convenience of the reader, Appendix~\ref{app:angular} contains expressions for angular momentum eigenstates along $x$ and $y$ axes in terms of standard angular momentum eigenstates along the $z$ axis.


\subsubsection[Spin-1/2]{\texorpdfstring{$j=1/2$}{Spin-1/2}}

\begin{figure}[t]
    \centering
    \includegraphics[width=0.6\columnwidth]{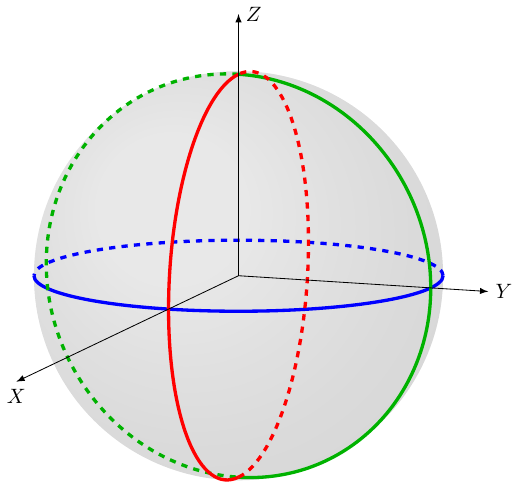}
    \caption{\label{fig:j=12}\textbf{No perfect quantum protractor for $j=1/2$.} States that are optimal quantum protractors with respect to axes $x$, $y$, and $z$ lie on great circles (green, red, and blue, respectively) around these axes on the Bloch sphere. States that lie on the intersection of two such great circles correspond to optimal quantum protractors of rank 2. Since the three great circles do not all intersect in one point, there is no perfect quantum protractor for spin-1/2 systems.}
\end{figure}

A general state $\ket{\psi}\in\H_{1/2}$, ignoring the irrelevant global phase, is given by:
\begin{equation}
\ket{\psi}=\sqrt{p_\uparrow} \ket{\uparrow}_z +\sqrt{p_{\downarrow}} e^{i\phi_{\downarrow}} \ket{\downarrow}_z.
\end{equation}
From Proposition~\ref{prop:optimal1}, we know that 
\begin{equation}
    \label{eq:1/2}
    p_\uparrow=p_{\downarrow}=\frac{1}{2}.  
\end{equation}
Using the above one finds
\begin{equation}
    |\braket{\psi|\!\uparrow}_x\!|^2=\frac{1+\cos\phi_\downarrow}{2},\quad |\braket{\psi|\!\uparrow}_y\!|^2=\frac{1-\sin\phi_\downarrow}{2}.
\end{equation}
Since it is impossible to make both of the above expressions simultaneously equal to $1/2$ (rendering the corresponding distributions uniform), for $j=1/2$ there does not exist a perfect quantum protractor. Note, however, that since each of the above two expressions can be separately equal to $1/2$ for two different $\phi_\downarrow$, there exist two optimal protractors of rank 2. We illustrate this in Fig.~\ref{fig:j=12}, where the condition of being an optimal quantum protractor with respect to a given axis can be understood geometrically. More generally, an optimal quantum protractor with respect to two arbitrary (not necessarily perpendicular) axes $\hat{\v{n}}$ and $\hat{\v{m}}$ is given by the eigenstates of $(\hat{\v{n}}\times \hat{\v{m}})\cdot{\v{J}}$, where $\times$ denotes the cross product.
 

\subsubsection[Spin-1]{\texorpdfstring{$j=1$}{Spin-1}}
\label{sec:spin-1}

Ignoring the irrelevant global phase and employing Proposition~\ref{prop:optimal1} makes a general state $\ket{\psi}\in\H_1$ take the following form:
\begin{equation}
  \label{eq:j=1}
\ket{\psi}=\frac{1}{\sqrt{3}}\left(e^{i\phi_\uparrow} \ket{\uparrow}_z +\ket{0}_z + e^{i\phi_{\downarrow}} \ket{\downarrow}_z\right).
\end{equation}
Using the above, we get
\begin{subequations}
\begin{align}    |\braket{\psi|\!\uparrow}_x\!|^2+|\braket{\psi|\!\downarrow}_x|^2&=\frac{2+\cos(\phi_\uparrow-\phi_\downarrow)}{3},\\ |\braket{\psi|\!\uparrow}_y|^2\!+|\braket{\psi|\!\downarrow}_y|^2&=\frac{2-\cos(\phi_\uparrow-\phi_\downarrow)}{3}.
\end{align}
\end{subequations}
Enforcing uniform distributions clearly means that one of the following two equations holds
\begin{equation}
    \phi_\downarrow=\phi_\uparrow+\pi/2,\qquad \phi_\downarrow=\phi_\uparrow-\pi/2,
\end{equation}
leading to
\begin{equation}                    |\braket{\psi|0}_x|^2=|\braket{\psi|0}_y|^2\!=\frac{1}{3}.    
\end{equation}
Moreover, in the first case, we get
\begin{align}    |\braket{\psi|\!\uparrow}_x\!|^2&=|\braket{\psi|\!\uparrow}_y|^2\!=\frac{2-\sqrt{2}\left(\cos\phi_\uparrow-\sin\phi_\uparrow\right)}{6},
\end{align}
and so uniform distributions along both $x$ and $y$ axes are obtained if and only if $\phi_\uparrow=\pi/4$ or $\phi_\uparrow=-3\pi/4$. In the second case, we get 
\begin{align}    |\braket{\psi|\!\uparrow}_x\!|^2&=|\braket{\psi|\!\downarrow}_y|^2\!=\frac{2-\sqrt{2}\left(\cos\phi_\uparrow+\sin\phi_\uparrow\right)}{6},
\end{align}
and so uniform distributions along both $x$ and $y$ axes are obtained if and only if $\phi_\uparrow=-\pi/4$ or $\phi_\uparrow=3\pi/4$. 

Therefore, for $j=1$ there exist four perfect quantum protractors given by Eq.~\eqref{eq:j=1} with the following phases (see Fig.~\ref{fig:j=1} for the visualisation of these solutions):
\begin{subequations}
\begin{align}
    \label{eq:angles_j1_1}
    \phi_{\uparrow}=\pm\frac{\pi}{4}\quad \mathrm{and} \quad &\phi_{\downarrow}=\pm\frac{3\pi}{4}, \\
    \label{eq:angles_j1_2}
    \phi_{\uparrow}=\pm\frac{3\pi}{4}\quad \mathrm{and} \quad & \phi_{\downarrow}=\pm\frac{\pi}{4}.
\end{align} 
\end{subequations}
Note that, interestingly, based on the above analysis the conditions for an optimal rank 2 protractor coincide with the conditions for a perfect protractor in the case of \mbox{$j=1$}.

\begin{figure}[t]
    \centering
    \includegraphics[width=0.9\columnwidth]{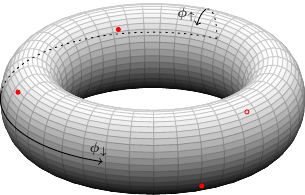}
    \caption{\label{fig:j=1}\textbf{Perfect quantum protractors for $j=1$.} Phases $\phi_\uparrow$ and $\phi_\downarrow$ in Eq.~\eqref{eq:j=1} that yield perfect quantum protractors are indicated by red circles (the empty circle indicates that the point lies on the surface of the torus invisible to the viewer). These states are also the only optimal quantum protractors of rank 2.}
\end{figure}


\subsubsection[Spin-3/2]{\texorpdfstring{$j=3/2$}{Spin-3/2}}
\label{sec:spin-32}

Following the same reasoning as before, the candidate state $\ket{\psi}$ for a perfect quantum protractor in $\H_{3/2}$ has the following form
\begin{align}
    \label{eq:state32}
    \ket{\psi}=&\frac{1}{2}\left(\ket{\Uparrow}_z+e^{i\phi_\uparrow}\ket{\uparrow}_z+e^{i\phi_\downarrow}\ket{\downarrow}_z+e^{i\phi_\Downarrow}\ket{\Downarrow}_z\right).
\end{align}
Using the above, we get
\begin{subequations}
\begin{align}
    \!\!|\!\braket{\psi|\!\Uparrow}_x\!|^2\!+\!|\!\braket{\psi|\!\downarrow}_x\!|^2=&\frac{1}{2}\!-\!\frac{\cos(\phi_\uparrow-\phi_\downarrow)+\cos\phi_\Downarrow}{4},\\ \!\!|\!\braket{\psi|\!\Uparrow}_y\!|^2\!+\!|\!\braket{\psi|\!\downarrow}_y\!|^2=&\frac{1}{2}\!+\!\frac{\sin(\phi_\uparrow-\phi_\downarrow)+\sin\phi_\Downarrow}{4}.
\end{align}
\end{subequations}
Thus, enforcing uniform distributions means that
\begin{equation}
    \label{eq:uparrow_const}
    \phi_\uparrow = \phi_\Downarrow+\phi_\downarrow+\pi.
\end{equation}
Using the above, we further get that $|\braket{\psi|\!\Uparrow}_x\!|=|\!\braket{\psi|\!\Downarrow}_y\!|$ if and only if
\begin{align}
    (\cos\phi_\Downarrow\!-\!\sin \phi_\Downarrow)(1\!+\!\sqrt{3}\sin\phi_\downarrow)=0.
\end{align}
The first parenthesis vanishes when $\phi_\Downarrow=\pi/4$ or $\phi_\Downarrow=-3\pi/4$, while the second one vanishes if \mbox{$\phi_\downarrow=-\arcsin(1/\sqrt{3})$} or $\phi_\downarrow=\arcsin(1/\sqrt{3})   -\pi$. For the first two options, enforcing $|\!\braket{\psi|\!\Uparrow}_x\! |=|\!\braket{\psi|\!\downarrow}_x\!|$ results in \mbox{$\sin\phi_\downarrow=1/\sqrt{3}$}, whereas enforcing the same for the other two options yields
\begin{equation}                
    \cos\phi_\Downarrow+\sin\phi_\Downarrow=0.
\end{equation}

Combing the above all together, we get the following eight candidate solutions:
\begin{equation}
\def\arraystretch{1.5}
\!\!\!\begin{array}{|c|ccc|}
    \hline
    \phi_\Downarrow & &\phi_\downarrow& \\\hline
    ~\pi/4~ & ~~\arcsin \frac{1}{\sqrt{3}} &\lor& \pi-\arcsin \frac{1}{\sqrt{3}}~~ \\ 
    \hline
    ~-3\pi/4~ & ~~\arcsin \frac{1}{\sqrt{3}} &\lor& \pi-\arcsin \frac{1}{\sqrt{3}}~~ \\ 
\hline
    ~ -\pi/4~ & ~~-\arcsin\frac{1}{\sqrt{3}} &\lor& \arcsin \frac{1}{\sqrt{3}}-\pi~~ \\ 
    \hline
    ~ 3\pi/4~ & ~~-\arcsin\frac{1}{\sqrt{3}} &\lor& \arcsin \frac{1}{\sqrt{3}}-\pi~~ \\ \hline
\end{array}~,
\end{equation}
where in each case $\phi_\uparrow$ is given by Eq.~\eqref{eq:uparrow_const}. Straightforward calculation shows that all these solutions actually yield perfect quantum protractors. We visualise these solutions in Fig.~\ref{fig:j=32}, where we also present all solutions for optimal protractors of rank 2.

\begin{figure}[t]
    \centering
    \includegraphics{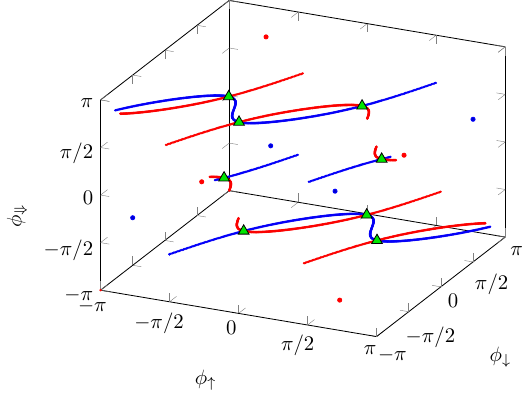}
\caption{\label{fig:j=32}\textbf{Perfect quantum protractors for $j=3/2$.} Phases $\phi_\uparrow$, $\phi_\downarrow$ and $\phi_\Downarrow$ in Eq.~\eqref{eq:state32} that yield optimal quantum protractors of rank 2 with respect to axes $x$ and $z$ ($y$~and~$z$) are indicated by blue (red) lines and points. The eight intersection points, indicated by green triangles, correspond to perfect quantum protractors.}
\end{figure}


\subsubsection[Spin-2]{\texorpdfstring{$j=2$}{Spin-2}}

As before, we start with the following $\ket{\psi}$ being the candidate for a perfect quantum protractor in $\H_2$:

\begin{align}
    \ket{\psi}=&\frac{1}{\sqrt{5}}\left(e^{i\phi_\Uparrow}\ket{\Uparrow}_z+e^{i\phi_\uparrow}\ket{\uparrow}_z+\ket{0}_z\right.\nonumber\\
    &\qquad\qquad\left.+e^{i\phi_\downarrow}\ket{\downarrow}_z+e^{i\phi_\Downarrow}\ket{\Downarrow}_z\right).
\end{align}
Using the above, we get
\begin{subequations}
\begin{align}
|\!\braket{\psi|\!\uparrow}_x\!|^2+|\!\braket{\psi|\!\downarrow}_x\!|^2=&\frac{2-a-b}{5},\\ |\!\braket{\psi|\!\uparrow}_y\!|^2+|\!\braket{\psi|\!\downarrow}_y\!|^2=&\frac{2-a+b}{5},
\end{align}
\end{subequations}
where
\begin{equation}
    a=\cos(\phi_\Uparrow-\phi_\Downarrow),\quad b=\cos(\phi_\uparrow-\phi_\downarrow).
\end{equation}
Since we want to enforce uniform distributions and as the only way for $a+b$ and $a-b$ to vanish simultaneously is $a=b=0$, we get four cases
\begin{equation}        \phi_\downarrow=\phi_\uparrow\pm\pi/2,\quad \phi_\Downarrow=\phi_\Uparrow\pm\pi/2,
\end{equation}
which we will denote by $C_{\pm\pm}$ and analyse one by one. 

For $C_{++}$ we have
\begin{align}
    \!\!|\!\braket{\psi|\!\downarrow}_x\!|^2\!-\!|\!\braket{\psi|\!\uparrow}_y\!|^2=&\frac{2}{5}\cos(\phi_\Uparrow-\phi_\uparrow),
\end{align}
leading to $\phi_\uparrow=\phi_\Uparrow\pm\pi/2$. In either case, additionally enforcing 
\begin{equation}
    \label{eq:enforcing1}
    |\!\braket{\psi|0}_x\! |=|\!\braket{\psi|0}_y\!|\quad\text{and}\quad |\!\braket{\psi|\!\Uparrow}_x\!|=|\!\braket{\psi|\!\Downarrow}_x\! |
\end{equation}
leads to a set of two contradictory equations
\begin{equation}
\label{eq:contradiction}
         \cos \phi_\Uparrow=\sin\phi_\Uparrow\quad\text{and}\quad \cos \phi_\Uparrow=-\sin\phi_\Uparrow. 
\end{equation}
For $C_{+-}$ we have
\begin{align}
    \!\!|\!\braket{\psi|\!\uparrow}_x\!|^2\!-\!|\!\braket{\psi|\!\downarrow}_y\!|^2=&\frac{2}{5}\sin(\phi_\Uparrow-\phi_\uparrow),
\end{align}
leading to $\phi_\uparrow=\phi_\Uparrow$ or $\phi_\uparrow=\phi_\Uparrow+\pi$. In either case, enforcing Eq.~\eqref{eq:enforcing1} leads again to a set of contradictory equations from Eq.~\eqref{eq:contradiction}. For $C_{-+}$ we have
\begin{align}
    \!\!|\!\braket{\psi|0}_x\!|=|\!\braket{\psi|0}_y\!|\quad \Longleftrightarrow \quad\sin\phi_\Uparrow=\cos \phi_\Uparrow,
\end{align}
leading to $\phi_\Uparrow=\pi/4$ or $\phi_\Uparrow=-3\pi/4$. In either case, enforcing 
\begin{equation}
    \label{eq:enforcing2} |\braket{\psi|\!\Uparrow}_x\!|=|\!\braket{\psi|\!\Uparrow}_y\!|\quad\text{and}\quad |\braket{\psi|\!\Uparrow}_x\!|=|\!\braket{\psi|\!\Downarrow}_y\!|  
\end{equation}
leads to a set of contradictory equations from Eq.~\eqref{eq:contradiction} with $\Uparrow$ subscript replaced by $\uparrow$. Finally, for $C_{--}$, we have 
\begin{align}
    \!\!|\!\braket{\psi|0}_x\!|=|\!\braket{\psi|0}_y\!| \quad\Longleftrightarrow\quad \sin\phi_\Uparrow=-\cos \phi_\Uparrow,
\end{align}
leading to $\phi_\Uparrow=-\pi/4$ or $\phi_\Uparrow=3\pi/4$. In either case, enforcing Eq.~\eqref{eq:enforcing2} leads to to a set of contradictory equations from Eq.~\eqref{eq:contradiction} with $\Uparrow$ subscript replaced by $\uparrow$. 

We thus conclude that it is impossible to find a state $\ket{\psi}$ with uniform distributions over the angular momentum eigenstates for three perpendicular axes, and thus for $j=2$ there does not exist a perfect quantum protractor. Of course, due to Corollary~\ref{cor:rank2}, there exist optimal quantum protractors of rank 2, with the following being one example:
\begin{subequations}
\begin{align}   \!\!\ket{\psi}&=\frac{e^{i\theta}\ket{\Uparrow}_{{z}}\!+\!\ket{\uparrow}_{{z}}\!+\!\ket{0}_{{z}}\!-\!\ket{\downarrow}_{{z}}\!+\!e^{i\theta}\ket{\Downarrow}_{{z}}}{\sqrt{5}}\\
    &=\frac{e^{i\phi}\ket{\Uparrow}_{{x}}\!-\!\ket{\uparrow}_{{x}}\!+\!e^{i(\phi+\theta)}\ket{0}_{{x}}\!+\!\ket{\downarrow}_{{x}}\!+\!e^{i\phi}\ket{\Downarrow}_{{x}}}{\sqrt{5}},
\end{align}
\end{subequations}
where
\begin{equation}
    \theta=\arctan\sqrt{\frac{5}{3}},\quad \phi=\arctan\sqrt{\frac{5}{27}}.
\end{equation}
and the solutions ound in Sec.~\ref{sec:spin-32}

\subsubsection[Spin-5/2]{\texorpdfstring{$j=5/2$}{Spin-5/2}}

Again, we start with the following candidate for a perfect quantum protractor in $\H_{5/2}$:
\begin{align}
    \ket{\psi}=&\frac{1}{\sqrt{6}}\left(\ket{\Uuparrow}_z+e^{i\phi_\Uparrow}\ket{\Uparrow}_z+e^{i\phi_\uparrow}\ket{\uparrow}_z\right.\nonumber\\
    &\qquad\left.+e^{i\phi_{\ssDdownarrow}}\ket{\Ddownarrow}_z+e^{i\phi_\Downarrow}\ket{\Downarrow}_z+e^{i\phi_\downarrow}\ket{\downarrow}_z\right).
\end{align}
To prove our point, we will need the following five quantities
\begin{subequations}
    \begin{align}
        A&:= |\!\braket{\psi|\!\Uuparrow}_x\!|^2+|\!\braket{\psi|\!\uparrow}_x\!|^2+|\!\braket{\psi|\!\Downarrow}_x\!|^2-\frac{1}{2},\\
        B&:=|\!\braket{\psi|\!\Uuparrow}_y\!|^2+|\!\braket{\psi|\!\uparrow}_y\!|^2+|\!\braket{\psi|\!\Downarrow}_y\!|^2-\frac{1}{2}, \\
        C&:=|\!\braket{\psi|\!\Uuparrow}_x\!|^2+|\!\braket{\psi|\!\Ddownarrow}_x\!|^2-\frac{1}{3}, \\
        D&:=|\!\braket{\psi|\!\Uparrow}_x\!|^2+|\!\braket{\psi|\!\Downarrow}_x\!|^2-\frac{1}{3}, \\
        E&:=C+|\!\braket{\psi|\!\Uuparrow}_y\!|^2+|\!\braket{\psi|\!\Ddownarrow}_y\!|^2-\frac{1}{3},
    \end{align}
\end{subequations}
which should all be equal to zero if we enforce uniform distributions. 

Direct calculation yields
\begin{align}
    E= \frac{\sqrt{5}}{24}\left(\cos \phi_{\Downarrow}+\cos(\phi_{\Uparrow}-\phi_{\sDdownarrow}) \right),
\end{align}
and so we get two options:
\begin{subequations}
    \begin{align}                       \phi_\Downarrow&=\pi+\phi_{\Uparrow}-\phi_{\sDdownarrow},\\
        \phi_{\Downarrow}&=\pi-\phi_{\Uparrow}+\phi_{\sDdownarrow}.\label{eq:j52_phase}
    \end{align}
\end{subequations}
Assuming the first of the above equations holds, one finds
    \begin{equation}
        A=-\frac{1}{6}\cos(\phi_\uparrow-\phi_\downarrow),\quad
        B=\frac{1}{6}\sin(\phi_\uparrow-\phi_\downarrow),
    \end{equation}
and thus we get a contradiction, since the above two expressions cannot vanish simultaneously. Therefore, we are forced to assume that Eq.~\eqref{eq:j52_phase} holds, which gives us
\begin{equation}
    A+iB = -\frac{1}{6}\left(e^{i(\phi_\downarrow-\phi_\uparrow)}+e^{i(\pi+2\phi_\Uparrow-\phi_{\ssDdownarrow})}+e^{i\phi_{\ssDdownarrow}}\right).
\end{equation}
Now, we observe that the above is a sum of three unit vectors on a complex plane, and such a sum can only vanish if these vectors form an equilateral triangle. As a result, one of the following two cases has to hold:
\begin{subequations}
\begin{align}
    \phi_{\sDdownarrow}+\frac{2\pi}{3}=\phi_\downarrow-\phi_\uparrow, \quad \phi_{\sDdownarrow}+\frac{4\pi}{3}=\pi+2\phi_\Uparrow-\phi_{\sDdownarrow},\\
    \phi_{\sDdownarrow}+\frac{4\pi}{3}=\phi_\downarrow-\phi_\uparrow, \quad \phi_{\sDdownarrow}+\frac{2\pi}{3}=\pi+2\phi_\Uparrow-\phi_{\sDdownarrow},
\end{align}    
\end{subequations}
which, after taking into account Eq.~\eqref{eq:j52_phase}, gives the following two $\mp$ options:
\begin{subequations}
    \begin{align}
        \phi_\downarrow&=\phi_\uparrow+\phi_\Uparrow\mp\frac{\pi}{2},\\
        \phi_\Downarrow&=\mp \frac{5\pi}{6},\\
        \phi_{\sDdownarrow}&=\phi_\Uparrow\mp\frac{11\pi}{6}.
    \end{align}
\end{subequations}

Finally, using these two sets of $\mp$ solutions to simplify expressions for $C$ and $D$, we get
\begin{subequations}
    \begin{align}
        C&=\frac{(\sqrt{5}-5\sqrt{3})\cos\phi_\uparrow\pm(\sqrt{15}+5)\sin\phi_\uparrow}{48\sqrt{2}},\\
        D&=\frac{(\sqrt{5}+3\sqrt{3})\cos\phi_\uparrow\pm(\sqrt{15}-3)\sin\phi_\uparrow}{48\sqrt{2}}.
    \end{align}
\end{subequations}
It is then clear that independently of the chosen $\mp$ solution, it is impossible to simultaneously make $C$ and $D$ vanish. We thus conclude that it is impossible to find a state $\ket{\psi}$ with uniform distributions over the angular momentum eigenstates for three perpendicular axes, and thus for $j=5/2$ there does not exist a perfect quantum protractor. 


\subsubsection[Spin-3 and higher angular momenta]{\texorpdfstring{$j=3$ and higher angular momenta}{Spin-3 and higher angular momenta}}

Resolving the problem of existence of perfect quantum protractors using direct analytical methods becomes less and less feasible with growing angular momentum~$j$, as the number of phase variables increases. One can then employ numerical optimisation methods to attack the problem for higher values of $j$. Namely, one can first use analytical methods to decrease the number of phase variables as much as possible, and then maximize some function $f$ of the remaining variables, with~$f$ chosen such that it attains its maximum for the perfect quantum protractor. One simple choice for that is to start with a state $\ket{\psi}$ from Proposition~\ref{prop:optimal1} with $\hat{\v{n}}=z$, remove the global phase (so that the number of phase $\phi_m$ variables is $2j$), and maximize
\begin{equation}
    f(\v{\phi})=H \left(\v{p}^{x}\right) + H\left(\v{p}^{y}\right) - 2\log(d).
\end{equation}
In the above $H(\cdot)$ denotes the Shannon entropy, whereas $\v{p}^{k}$ is the probability distribution over measurements of angular momentum along axes \mbox{$k\in\{x,y\}$}. Clearly, since the maximum value of Shannon entropy is $\log d$, and for a perfect quantum protractor one gets uniform distributions for measurements along all three perpendicular axes, the function $f$ attains its maximum equal to 0 only for perfect quantum protractors. 

Using this method, we were able to find a perfect quantum protractor for $j=3$, i.e., we recovered the analytical expressions for the phases from the numerical solutions for which $f$ was equal to 0. It can be verified analytically that the following state,
\begin{align}
    \label{eq:state3}
    \ket{\psi}=&\frac{1}{\sqrt{7}}\left(e^{i\phi_{\ssUuparrow}}\ket{\Uuparrow}_z+e^{i\phi_\Uparrow}\ket{\Uparrow}_z+e^{i\phi_\uparrow}\ket{\uparrow}_z+\ket{0}\right.\nonumber\\
    &\quad\qquad\left.+e^{i\phi_{\ssDdownarrow}}\ket{\Ddownarrow}_z+e^{i\phi_\Downarrow}\ket{\Downarrow}_z+e^{i\phi_\downarrow}\ket{\downarrow}_z\right),
\end{align}
with phases
\begin{equation}
\begin{split}
    \label{eq:angles3}
    \phi_{\sUuparrow}=-\theta-\frac{\pi}{4},\quad \phi_\Uparrow=\pi/4,\quad \phi_\uparrow=\theta+\frac{3\pi}{4},\\
    \phi_{\sDdownarrow}=\theta-\frac{3\pi}{4},\quad \phi_\Downarrow=-\pi/4,\quad \phi_\uparrow=-\theta+\frac{\pi}{4},
\end{split}
\end{equation}
where $\theta=\frac{1}{2}\arctan\sqrt{15}$, is indeed a perfect quantum protractor (but most probably not the only one).

Moreover, for $j=7/2$ the numerical optimization gave the value of $f$ to be zero up to numerical precision, which strongly suggests that a perfect quantum protractor exists in $\H_{7/2}$ (see Appendix~\ref{app:numerics} for details). For higher values of $j$, this simple numerical optimisation procedure does not find phases for which $f$ is equal to zero. This, however, may be inconclusive, as the number of variables may be too large to efficiently find the global maximum. 


\section{Properties} 
\label{sec:properties}

In this section we will argue that perfect quantum protractors are pure states that maximise the uncertainty of angular momentum measurements. Thus, they can be considered as the opposite of the much more studied states that minimise angular momentum uncertainty~\cite{dammeier2015uncertainty}, such as spin coherent states~\cite{radcliffe1971some} or spin intelligent states~\cite{aragone1976intelligent}. As we shall see, this property of maximising uncertainty renders angular momentum certainty relations trivial, and makes perfect quantum protractors optimal resources for particular metrological tasks.


\subsection{Entropic uncertainties}

Let us start with information-theoretic (entropic) quantifiers of uncertainty. A very general family of such uncertainty measures, based on a minimal set of axioms, is given by R\'{e}nyi entropies~\cite{renyi1961measures}. These are defined by
\begin{equation}
    H_\alpha(\v{p}):=\frac{1}{1-\alpha}\log\left(\sum_i p_i^\alpha\right),
\end{equation}
where $\alpha\geq 0$\footnote{Importantly, $\lim_{\alpha\to 1} H_\alpha$ is the previously introduced and most widely used entropic quantity -- the Shannon entropy~$H$.} and $\v{p}$ denotes the probability vector of observing different events. Here, these events will correspond to measuring different values of angular momentum along different axes for a state $\ket{\psi}$, and so we introduce probability distributions $\v{p}^x$, $\v{p}^y$ and $\v{p}^z$, whose entries are given by \mbox{$|\!\braket{\psi|j,m}_x\!|^2$}, \mbox{$|\!\braket{\psi|j,m}_y\!|^2$} and \mbox{$|\!\braket{\psi|j,m}_z\!|^2$}, respectively. 

Now, it is easy to see that, independently of~$\alpha$, R\'enyi entropies achieve their maxima, equal to $\log (2j+1)$, only for a uniform distribution over angular momentum eigenstates. Thus, the following very general entropic quantifier of total uncertainty about angular momenta measurements is trivially upper bounded,
\begin{equation}
    H_{\alpha_x}\left(\v{p}^{x}\right)+H_{\alpha_y}\left(\v{p}^{y}\right)+H_{\alpha_z}\left(\v{p}^{z}\right)\leq 3\log (2j+1),
\end{equation}
for all $\alpha_k\geq 0$ and $k\in\{x,y,z\}$. Although physicists are mainly preoccupied with the lower bounds for such uncertainty quantifiers (with the famous Maassen-Uffink entropic uncertainty relation~\cite{maassen1988generalized} being a prime example), the non-trivial upper bounds were also studied under the name of certainty relations~\cite{sanchez1993entropic,sanchez1995improved,puchala2015certainty}. For example, in Ref.~\cite{sanchez1993entropic} it was shown that for a spin $j=1/2$ system one has
\begin{equation}
    \label{eq:certainty_ex}
    H(\v{p}^{x})+H(\v{p}^{y})+H(\v{p}^{z})\leq \frac{3\log 6\!-\!\sqrt{3}\log(2\!+\!\sqrt{3})}{2}.
\end{equation}

However, the existence of a perfect protractor for a spin-$j$ system means that there is no non-trivial upper bound on the angular momentum uncertainty. This is because Proposition~\ref{prop:optimal1} requires perfect quantum protractors to be given by a uniform superposition over all $(2j+1)$ angular momentum eigenstates for all three axes, and so such states attain the trivial upper bound of \mbox{$3 \log (2j+1)$}. Thus, based on the results presented in this work, we know that there are no equivalents of Eq.~\eqref{eq:certainty_ex} for spin-$j$ systems when \mbox{$j\in\{1,3/2,3,7/2\}$}. More generally, we can summarise it as follows.

\begin{corollary}
    If a perfect quantum protractor exists in~$\H_j$, then there is no non-trivial entropic certainty relation for angular momentum operators, since it saturates the trivial bound,
    \begin{equation}
        H(\v{p}^{x})+H(\v{p}^{y})+H(\v{p}^{z}) = 3\log(2j+1).
    \end{equation}    
\end{corollary}


\subsection{Variance-based uncertainties}
\label{sec:variances}

Let us now proceed to more traditional quantifiers of uncertainty based on the variance,
\begin{equation}
    \label{eq:variance}
    \sigma^2_k := \langle J_k^2\rangle -\langle J_k\rangle^2,
\end{equation}
where we omit the dependence of $\sigma^2_k$ on $\ket{\psi}$ and use a shorthand notation with $\langle A\rangle$ denoting the expectation value of $A$ in a state $\ket{\psi}$. We can then consider three Pythagorean means (arithmetic, geometric and harmonic) of angular momentum variances along three perpendicular axes:
\begin{subequations}
\begin{align}
    \bar{\sigma}^2_A&:=\frac{\sigma^2_x+\sigma^2_y+\sigma^2_z}{3},\\
    \bar{\sigma}^2_G&:=\sqrt[3]{\sigma^2_x\sigma^2_y\sigma^2_z},\\
    \bar{\sigma}^2_H&:=\left(\frac{(\sigma^2_x)^{-1}+(\sigma^2_y)^{-1}+(\sigma^2_z)^{-1}}{3}\right)^{-1}.
\end{align}
\end{subequations}
Crucially, for a perfect quantum protractor, variances for all three axes are the same,
\begin{align}
    \!\!\sigma^2_k&=\frac{1}{2j\!+\!1}\left(\sum_{m=-j}^j \!\! m^2 - \left(\sum_{m=-j}^j \!\! m\right)^{2}\right)=\frac{j(j\!+\!1)}{3},
\end{align}
and so all Pythagorean means are equal to \mbox{$j(j+1)/3$}.

For the arithmetic mean, let us note that
\begin{equation}
    \label{eq:var_sum}
    \bar{\sigma}^2_A \leq \frac{\braket{J_x^2}+\braket{J_y^2}+\braket{J_z^2}}{3}=\frac{\braket{J^2}}{3}=\frac{j(j+1)}{3},
\end{equation}
and this upper bound is attained by states satisfying
\begin{equation}
    \braket{J_x}=\braket{J_y}=\braket{J_z}=0.
\end{equation}
Such states are known in the literature as spin 1-anticoherent states~\cite{zimba2006anticoherent,goldberg2018quantum}, and perfect quantum protractors form particular examples of such states. Note also that perfect quantum protractors have equal angular momentum variances along three perpendicular axes, $x$, $y$, and $z$, but not necessarily along any axis as spin 2-anticoherent states have (see Sec.~\ref{sec:conclusions} for more discussion on that subject).

For the geometric mean, we have
\begin{align}
    \left(\bar{\sigma}^2_G\right)^3 \leq& \braket{J_x^2}\braket{J_y^2}\braket{J_z^2}\nonumber\\
    =&\braket{J_x^2}\braket{J_y^2}\left(j(j+1)-\braket{J_x^2}-\braket{J_y^2}\right).
\end{align}
Thus, to upper bound it, we can consider the following function of two variables,
\begin{equation}
    g(x,y):= x y (j(j+1)-x-y),
\end{equation}
and find its maximum within the region specified by the following constraints:
\begin{subequations}
\begin{align}
    \label{eq:constr1}
    0\leq x,y &\leq j^2,\\
    x+y&\leq j(j+1).\label{eq:constr2}
\end{align}
\end{subequations}
Direct calculation of partial derivatives yields the following maximum
\begin{equation}
\label{eq:max_geo}
    g\left(x=\frac{j(j+1)}{3},y=\frac{j(j+1)}{3}\right)=\left(\frac{j(j+1)}{3}\right)^3.
\end{equation}
To be sure that this is indeed an upper bound, we also need to verify the value of $g$ at the boundaries specified by Eqs.~\eqref{eq:constr1}-\eqref{eq:constr2}. Clearly, for $x=0$, $y=0$ and \mbox{$x+y=j(j+1)$}, we get $g=0$, and so the only non-trivial boundary one has to consider is given by $x=j^2$ and $y\in[0,j]$ (the case with inverted $x$ and $y$ is analogous). This yields
\begin{equation}
    g(x=j^2,y)=j^2 y (j-y),
\end{equation}
which takes its maximum at $y=j/2$, resulting in 
\begin{equation}
    g\left(x=j^2,y=\frac{j}{2}\right)=\frac{j^4}{4}.
\end{equation}
Since the above value is smaller than the one from Eq.~\eqref{eq:max_geo}, we conclude that
\begin{equation}
    \bar{\sigma}^2_G \leq \frac{j(j+1)}{3},
\end{equation}
and that the perfect quantum protractor attains the above upper bound.

Finally, for the harmonic mean, we have
\begin{align}
    3\left(\bar{\sigma}^2_H\right)^{-1}&=(\sigma^2_x)^{-1}+(\sigma^2_y)^{-1}+(\sigma^2_z)^{-1}\nonumber\\
    &=\frac{1}{{\sigma}^2_x}+\frac{1}{{\sigma}^2_y}+\frac{1}{j(j+1)-{\sigma}^2_x-{\sigma}^2_y}.
\end{align}
To upper bound it, we can then consider the following function,
\begin{equation}
    h(x,y)=\frac{1}{x}+\frac{1}{y}+\frac{1}{j(j+1)-x-y},
\end{equation}
and, as before, find its maximum value subject to constraints from Eqs.~\eqref{eq:constr1}-\eqref{eq:constr2}. Analogous and straightforward calculations lead to the same result, i.e., that the maximum is achieved at \mbox{$x=y=j(j+1)/3$}. Therefore, the harmonic mean is upper bounded by
\begin{equation}
    \label{eq:harmonic_ineq}
    \bar{\sigma}^2_H \leq \frac{j(j+1)}{3},
\end{equation}
and perfect quantum protractors saturate the above bound. 

We thus conclude that perfect quantum protractors maximise all Pythagorean means of variance-based uncertainties about measurements of angular momenta along three perpendicular axes. Thus, these are the states of complete knowledge (since they are pure states) that nevertheless yield maximal uncertainty about angular momenta. As we shall see in the next section, this makes them a valuable resource for quantum metrology.


\subsection{Entanglement}
\label{sec:entanglement}

In Sec.~\ref{sec:existence}, we have proved that a perfect quantum protractor for a $d$-dimensional system can only exist for systems carrying an irrep $j$. This means that it is impossible to obtain a perfect quantum protractor, while combining two (or more) systems carrying irreps $j_1>0$ and $j_2>0$ (with dimensions $d_k=2j_k+1$). The reason is that a perfect quantum protractor, by definition, would need to distinguish \mbox{$d_1d_2=(2j_1+1)(2j_2+1)$} angles; but we know that the highest irrep appearing in the decomposition of $\H_{j_1}\otimes\H_{j_2}$ is $j_{\max}=j_1+j_2$, and so the best we can do is to distinguish $2j_{\max}+1<d_1d_2$ angles. 

Nevertheless, one may focus on perfect quantum protractors living in the highest irrep subspace $\H_{j_{\max}}$ of the joint Hilbert space of a few subsystems, and investigate the resulting entanglement between the subsystems. To clarify this idea, let us start with the simplest example of two spin-1/2 systems, whose total Hilbert space $\H=\H_{1/2}\otimes\H_{1/2}$ splits into a direct sum of $\H_0$ and $\H_1$. Now, we know that in order to get the best quantum protractor $\ket{\psi}\in\H$, it should live in the subspace corresponding to the highest irrep, so that $\ket{\psi}\in \H_1$. We can then take solutions for perfect quantum protractors of spin-1 system from Sec.~\ref{sec:spin-1} and write them explicitly in the basis of two spin-1/2 systems:
\begin{equation}
    \!\!\!\!\ket{\psi}=\frac{1}{\sqrt{3}} \left(\!e^{i\phi_\uparrow}\ket{\uparrow\uparrow}+\frac{1}{\sqrt{2}}(\ket{\uparrow\downarrow}+\ket{\downarrow\uparrow})+e^{i\phi_\downarrow}\ket{\downarrow\downarrow}\!\right),\!
\end{equation}
with the angles $\phi_\uparrow$ and $\phi_\downarrow$ given by Eqs.~\eqref{eq:angles_j1_1}-\eqref{eq:angles_j1_2}. It is a straightforward calculation to show that
\begin{equation}
    \ket{\psi}=\frac{1}{\sqrt{2}} \left[\ket{\phi_\uparrow}\otimes\ket{\uparrow}+\ket{\phi_\downarrow}\otimes\ket{\downarrow} \right],
\end{equation}
where
\begin{subequations}
    \begin{align}
        \ket{\phi_\uparrow}&:=\sqrt{\frac{2}{3}}e^{i\phi_\uparrow}\ket{\uparrow}+\sqrt{\frac{1}{3}}\ket{\downarrow},\\
        \ket{\phi_\downarrow}&:=\sqrt{\frac{1}{3}}\ket{\uparrow}+\sqrt{\frac{2}{3}}e^{i\phi_\downarrow}\ket{\downarrow},
    \end{align}
\end{subequations}
are orthogonal states. Thus, a $j=1$ perfect quantum protractor $\ket{\psi}$ constructed from two spin-1/2 subsystems is a maximally entangled state of these two subsystems. At the same time, it is an easy exercise to verify that for a product state of two spin-1/2 systems, it is impossible to obtain a perfect quantum protractor (the best one can do is the same as with a single spin-1/2 system).

The next perfect quantum protractor $\ket{\psi}$ that we found in this paper lives in $\H_{3/2}$, which is the highest irrep in the decomposition of the total Hilbert space of three spin-1/2 subsystems, or spin-1 and spin-1/2 subsystems. One can then again investigate the entanglement properties of these three or two subsystems in a state $\ket{\psi}$. Here, we will analyse in detail the first (harder) case, and just state the result for the second case, which can be obtained in a very similar manner. Using Clebsch-Gordan coefficients, we have
\begin{subequations}
\begin{align}   
    \ket{\Uparrow}=\ket{\uparrow\uparrow\uparrow},\qquad \ket{\Downarrow}=\ket{\downarrow\downarrow\downarrow},\\
    \ket{\uparrow}=\frac{1}{\sqrt{3}} (\ket{\uparrow\uparrow\downarrow}+\ket{\uparrow\downarrow\uparrow}+\ket{\downarrow\uparrow\uparrow}),\\
    \ket{\downarrow}=\frac{1}{\sqrt{3}} (\ket{\downarrow\downarrow\uparrow}+\ket{\downarrow\uparrow\downarrow}+\ket{\uparrow\downarrow\downarrow}),
\end{align}
\end{subequations}
where the states on the left hand sides of the above equations describe the total spin-3/2 system, and the states on the right hand sides of these equations describe three spin-1/2 subsystems. We can then combine the above with a general formula for spin-3/2 perfect quantum protractor, Eq.~\eqref{eq:state32}, to arrive at
\begin{align}
    \ket{\psi} =& \frac{1}{2\sqrt{3}}\left(\sqrt{3}\ket{\uparrow\uparrow\uparrow} +e^{i\phi_\uparrow}(\ket{\uparrow\uparrow\downarrow}+\ket{\uparrow\downarrow\uparrow}+\ket{\downarrow\uparrow\uparrow} )\right.\nonumber\\
    &\left. e^{i\phi_\downarrow}(\ket{\downarrow\downarrow\uparrow}+\ket{\downarrow\uparrow\downarrow}+\ket{\uparrow\downarrow\downarrow})+e^{i\phi_\Downarrow}\sqrt{3}\ket{\downarrow\downarrow\downarrow}\right),
\end{align}
Using eight sets of angles found in Sec.~\ref{sec:spin-32}, we get eight spin-3/2 perfect quantum protractors expressed in the basis of three spin-1/2 subsystems. The simplest way to analyse their entanglement properties is to notice that we deal with pure states for which partial traces over any two spin-1/2 subsystems leave us in a maximally mixed state of a single spin-1/2 subsystem. It is then known~\cite{scott2004multipartite}, that such a state is a GHZ state, i.e., up to local unitaries it is given by $(\ket{\uparrow\uparrow\uparrow}+\ket{\downarrow\downarrow\downarrow})/\sqrt{2}$ which, according to many multipartite entanglement measures, is regarded as a maximally entangled state of three subsystems. And for the case of combining spin-1 and spin-1/2 subsystems, one can show that a perfect quantum protractor is given by a maximally entangled bipartite state.

Finally, the last perfect quantum protractor for which we have analytic expressions lives in $\H_3$. There are many combinations of subsystems for which this subspace appears as the highest irrep, e.g., two spin-3/2 subsystems, three spin-1 subsystems, six spin-1/2 subsystems, one spin-1 and one spin-2 subsystems, etc. Here, we will limit our considerations only to the first case of $\H_{3/2}\otimes\H_{3/2}$, as our aim is not to provide a full analysis, but rather to show that, despite what we have seen so far, perfect quantum protractors not always need to correspond to maximally entangled states. As before, using Clebsch-Gordan coefficients, we can rewrite the spin-3 perfect quantum protractor from Eq.~\eqref{eq:state3} (with angles given by Eq.~\eqref{eq:angles3}) in the basis of two spin-3/2 subsystems, and then trace out one of the subsystems. The resulting mixed state is given by
\begin{equation}
    \rho=\begin{pmatrix}
        \frac{1}{4}& a e^{-3i\pi/4} & \frac{ai}{\sqrt{2}} & 0 \\
        a e^{3i\pi/4} &\frac{1}{4} &0 & \frac{ai}{\sqrt{2}}\\
        -\frac{ai}{\sqrt{2}} & 0 & \frac{1}{4} & a e^{i\pi/4}\\
        0 & -\frac{ai}{\sqrt{2}} & a e^{-i\pi/4} & \frac{1}{4}
    \end{pmatrix},
\end{equation}
where
\begin{equation}
    a=\frac{3\sqrt{6}+2\sqrt{10}}{70}.
\end{equation}
The above is clearly not a maximally mixed state, and so, while the two spin-3/2 subsystems are entangled, they are not maximally entangled. 


\section{Metrological applications}
\label{sec:metrology}

A paradigmatic example of a metrological task in the quantum realm is given by a phase estimation problem~\cite{giovannetti2011advances}. In this setting, a system initially prepared in a generally mixed state $\rho$ undergoes a unitary evolution that transforms it as
\begin{equation}
    \label{eq:phase_estim}
    \rho \rightarrow \rho_{\theta} := e^{-iH \theta} \rho e^{iH \theta},
\end{equation}
where $H$ is a known self-adjoint operator and $\theta$ is a real number corresponding to an unknown phase shift. The goal is to find the estimate $\hat{\theta}$ of $\theta$ by preparing the system in the optimal probe state $\rho$, allowing it to undergo the unknown phase shift $\theta$ and performing a generalised quantum measurement specified by POVM elements $\{\Pi_{\hat{\theta}}\}$. Finally, upon observing outcome $\hat{\theta}$, which happens with probability
\begin{equation}
    p(\hat{\theta}|\theta) = \tr{\rho_{\theta}\Pi_{\hat{\theta}}},
\end{equation}
one guesses that $\theta=\hat{\theta}$. Of course, the choice of the optimal strategy (i.e., $\rho$ and $\{\Pi_{\hat{\theta}}\}$) depends on the measure of success one chooses, as well as on the prior knowledge about the distribution of $\theta$~\cite{demkowicz2011optimal,hall2012heisenberg}.

Here, we will consider the case when $H$ in Eq.~\eqref{eq:phase_estim} is given by an angular momentum operator $J_k$. Thus, the transformation corresponds either to a physical rotation of the system around the axis $k$ by an angle $\theta$, or to a time evolution generated by a Hamiltonian given by the angular momentum along the axis $k$ (e.g., precession in a magnetic field). Estimating the phase can then provide information about the angle of rotation or the strength of the magnetic field, so the probe state $\rho$ can be seen as a protractor or a magnetometer.

If we only cared about estimating $\theta$ for a single $J_k$, we could simply use the known results on quantum phase estimation. Instead here, we will focus on a very particular case of a multiparameter estimation problem, involving three different generators: $J_x$, $J_y$ and $J_z$ (for a related problem of finding optimal states for aligning Cartesian reference frames see Ref.~\cite{kolenderski2008optimal}). Namely, imagine that one knows that an unknown phase shift of $\theta_x$, $\theta_y$ or $\theta_z$ will be generated by either $J_x$, $J_y$ or $J_z$. However, one does not \emph{a priori} know which generator will be used (a situation referred to as a random sensing scenario in Ref.~\cite{albarelli2022probe}), so one needs to prepare a probe state $\rho$ that will be able to detect all these different phase shifts. Then, before the measurement is performed, one obtains the information about the generator (i.e., the axis of rotation), and can choose an appropriate measurement. As we will now argue, the optimal pure probe state in such a scenario is given by a perfect quantum protractor.


\subsection{Standard setup}

Given $n$ copies of a pure state $\ket{\psi}$, the quantum Cram{\'e}r-Rao bound sets the following (attainable) lower bound on the inaccuracy of estimating the phase~$\theta_k$ generated by $J_k$~\cite{helstrom1969quantum}:
\begin{equation}
    \Delta\theta_k^2\geq \frac{1}{4n\sigma^2_k},
\end{equation}
where $\Delta\theta_k^2$ is the mean square error of estimating $\theta_k$ and $\sigma_k^2$ is the variance of $J_k$ in a state $\ket{\psi}$ as defined in Eq.~\eqref{eq:variance}. Now, if states $\ket{\psi}$ are used to estimate the phases $\theta_x$, $\theta_y$ and $\theta_z$ generated by $J_x$, $J_y$ and $J_z$ with equal probability 1/3, then the expected inaccuracy is lower bounded by
\begin{equation}
    \frac{\Delta\theta_x^2+\Delta\theta_y^2+\Delta\theta_z^2}{3}\geq \frac{1}{4n\bar{\sigma}^2_H}\geq \frac{3}{4n j (j+1)},
\end{equation}
where the second inequality comes from Eq.~\eqref{eq:harmonic_ineq}. As this inequality is saturated by $\ket{\psi}$ being a perfect quantum protractor, we conclude that such states minimise the inaccuracy of estimation and are thus optimal for the considered metrological task.


\subsection{Discrete setup}

The existence of perfect quantum protractors also affects a discrete version of the estimation problem studied recently in Ref.~\cite{rzkadkowski2017discrete}. Discrete here means that the phase is not arbitrary, but takes one of the $n$ values $\theta_l=2\pi l/n$ for $l\in\{0,1,\dots,n-1\}$ with equal probability $1/n$. Fixing for a while a single generator given by $J_k$, the problem of phase estimation can then actually be seen as a problem of state discrimination between $n$ states,
\begin{equation}
    \rho_{l} = e^{-iJ_k\theta_l} \rho e^{iJ_k\theta_l}, 
\end{equation}
where $\rho$ is the initial probe state of size $d$. The probability of success is then given by
\begin{equation}
  \!\!\!  p_{\mathrm{succ}}\left(\rho,\{\Pi_l\}\right) = \frac{1}{n}\sum_{l=1}^n  \tr{\rho_{l} \Pi_{l}},
\end{equation}
where $\{\Pi_l\}$ denote POVM elements. Following Ref.~\cite{rzkadkowski2017discrete}, one can show that optimising over all states and over all POVMs, this probability is given by
\begin{equation}
    \! p_{\mathrm{succ}}^{\mathrm{opt}}:=\!\max_{\rho,\{\Pi_l\}}p_{\mathrm{succ}}\left(\rho,\{\Pi_l\}\right)=\!\left\{ \begin{array}{cc}
         d/n& \mathrm{for~} n \geq d,\\
         1&  \mathrm{for~}n<d.
    \end{array}\right.
\end{equation}

Crucially, when $n\geq d$, this optimal value can be achieved by the optimal quantum protractor with respect to the axis $k$. One simply prepares a pure probe state $\rho=\ketbra{\psi}{\psi}$ with $\ket{\psi}$ given by Eq.~\eqref{eq:optimal1} (with the axis $\hat{\v{n}}$ specified by $k$) and POVM elements given by
\begin{equation}
    \Pi_{l} = 
         \frac{d}{n}e^{-i J_k \theta_l} \ketbra{\psi}{\psi} e^{i J_k \theta_l}.
\end{equation}
The above operators are obviously positive and it is a straightforward calculation to show that
\begin{equation}
    \sum_l \Pi_{l}=\iden.    
\end{equation}
We then get
\begin{equation}
    p_{\mathrm{succ}} = \frac{1}{n}\sum_{l=1}^n  \tr{\rho_{l} \Pi_{l}} = \frac{d}{n},
\end{equation}
where each term in the sum is equal to $d/n$, so the success probability is independent of the actual phase. Since perfect quantum protractors have the properties used above simultaneously for three different generators given by $J_x$, $J_y$ and $J_z$, we can thus conclude that such states can be employed as optimal probe states simultaneously maximising the probability of success for three estimation problems of discrete phases generated by $J_x$, $J_y$ and~$J_z$.


\section{Experiment}
\label{sec:experiment}

\subsection{Experimental setup and sequence}

To experimentally implement perfect quantum protractors, an atomic vapour of rubidium-87 is used. The lowest energy states of the atoms are two hyperfine levels with the total angular momentum $j=1$ and $j=2$.\footnote{Conventionally, these states are indicated by $F$. However, to keep consistency with the previous discussion, here they are denoted by $j$.} At 50$^\circ$C, the vapour concentration is 10$^{11}$~atoms/cm$^3$, which indicates that individual interactions between atoms, being their point-like collisions, are rare (with characteristic times on the order of hundreds of milliseconds). In turn, the quantum state of the system can be described using a single-atom/collective density matrix~\cite{Auzinsh2010OpticallyInteractions}. 

The rubidium atoms are contained in a 3-cm in diameter paraffin-coated vapour cell (the total number of atoms in the cell is $1.4\times 10^{12}$~atoms), which is placed inside a magnetic field consisting of three layers of mu-metal and a single, innermost ferrite layer. The shield of a shielding factor 10$^{6}$ provides stable magnetic conditions, reducing external, uncontrollable magnetic fields. Apart from the cell, the shield additionally houses a set of magnetic-field coils, which enables compensation of residual magnetic fields, as well as generation of homogeneous magnetic-field pulses in the ${x}$, ${y}$, and ${z}$ directions. The light used to illuminate rubidium atoms is provided by three diode lasers. The pump laser is used to create the initial state in the atoms and the probe laser is used to measure/reconstruct the state. Additionally, the repump laser enables depletion of the $j=2$ level so that the whole system effectively reduces to a spin-1 system. In general, the experimental sequence is divided into three stages: (1) generation of a desired quantum state, which involves optical pumping with the pump tuned to the $j=1$ state along with the depletion of the $j=2$ state with the repump, (2) manipulation of the state with magnetic fields, and (3) measurements of the state. The simplified scheme of the experimental setup is shown in Fig.~\ref{fig:setup} and described in more detail in Ref.~\cite{Kopciuch2024optimized}.

\begin{figure}[t]
    \centering
    \includegraphics[width=0.9\columnwidth]{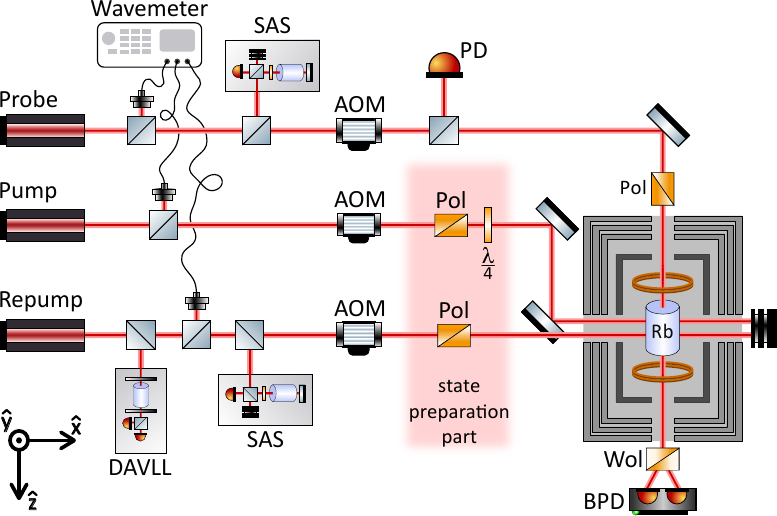}
    \caption{\textbf{Experimental setup.} SAS -- saturation absorption spectroscopy, DAVLL -- dichroic atomic vapour laser lock, AOM -- acousto-optic modulator, PD -- photodiode, Pol -- Glan-Thompson polariser, $\lambda/4$ -- quarter-wave plate, $^{87}$Rb -- paraffin-coated vapour cell filled with $^{87}$Rb, Wol -- Wollaston prism, BPD - balanced photodetector.}
    \label{fig:setup}
\end{figure}

The experimental goal of this work is to generate the ideal quantum protractor state $\ket{\psi}$ and demonstrate that it is the optimal state for measurements of rotations around the $x$, $y$, and $z$ axes, as well as for the determination of magnetic-field strength along these axes. Without loss of generality, we use for that the state derived in Sec.~\ref{sec:spin-1},
\begin{equation}
    \label{eq:exp_protractor}
    \ket{\psi}=\frac{1}{\sqrt{3}}\left(e^{i\frac{3\pi}{4}}\ket{\uparrow}_z+\ket{0}_z+e^{i\frac{\pi}{4}}\ket{\downarrow}_z\right).
\end{equation}
Such a state can be generated from the $\ket{0}_z$ state through the application of a sequence of geometric rotations:
\begin{equation}
    \label{eq:generate_psi}
    \left\vert \psi \right\rangle = R_z\left(-\dfrac{\pi}{4}\right)R_x\left( -\cos^{-1}\left( \dfrac{1}{\sqrt{3}} \right) \right) \left\vert 0 \right\rangle_z.
\end{equation}
Experimentally, the state $\left\vert \psi \right\rangle$ can be generated with $z$-polarized light propagating in the $x$ direction that is tuned to the $j=1\rightarrow j'=1$ transition, which generates the $|0\rangle_z$ state, and non-oscillating magnetic-field pulses applied along the $x$ and $z$-axis. 

As shown in Refs.~\cite{Kopciuch2022, Kopciuch2024optimized}, the time-dependent
polarization rotation of the probe light is given by\footnote{Here, we replaced originally used observables $\langle\alpha_R\rangle$, $\langle\alpha_I\rangle$, and $\langle\beta\rangle$ with the ones used in this work: $\langle M_1^z\rangle$, $\langle M_2^z\rangle$, and $\langle M_3^z\rangle$.}
\begin{align}
     \label{eq:signal}
     \delta\alpha(t,\Delta)=&~\eta(\Delta)e^{-\gamma_1t}\left[
\expval{M^z_1} \sin(2\Omega_Lt)\!+\!\expval{M^z_2}
\cos(2\Omega_Lt)\right]\nonumber\\&~-\zeta(\Delta)e^{-\gamma_2t}\expval{M^z_3},
\end{align}
where $\eta$ and $\zeta$ are scaling constants that depend on the
probe light detuning $\Delta$, $\Omega_L$ is the Larmor-precession
frequency of the atoms, $\gamma_1$ and $\gamma_2$ are the relaxation
rates of coherences and populations, respectively (for more details see
Ref.~\cite{Kopciuch2024optimized}). For the spin-1 system probed with the light tuned to the $j=1\rightarrow j'=1$ transition, the observables governing the signal amplitude are given by:
\begin{subequations}
     \begin{align}
         M^z_1 =&  J_x^2-J_y^2,\label{eq:obs1}\\
         M^z_2 =&  J_xJ_y+J_yJ_x,\label{eq:obs2}\\
         M^z_3 =& J_z.\label{eq:obs3}
     \end{align}
\end{subequations}
Hence, by measuring and analysing the time-dependent rotation $\delta\alpha$,
we can extract the expectation values $\langle M^z_i\rangle$. As explained in the following section, these values contain information about the rotation $\theta_z$ applied around the $z$ axis. To measure rotations around the other two axes by the angles $\theta_x$ or
$\theta_y$, we utilize magnetic-field pulses that effectively transform
the $z$ axis into the $x$ or $y$ axis before the actual
polarization-rotation measurements~\cite{Kopciuch2024optimized}.

\subsection{Optimality}

Using the properties of the angular momentum operators, one can show that
\begin{subequations}
\begin{align}
     e^{iJ_z\theta_z}J_xe^{-iJ_z \theta_z} &= \cos\theta_z J_x -
\sin\theta_z J_y,\\
     e^{iJ_z\theta_z}J_ye^{-iJ_z \theta_z} &= \sin\theta_z J_x +
\cos\theta_z J_y,
\end{align}
\end{subequations}
When one combines the equations with Eqs.~\eqref{eq:obs1}-\eqref{eq:obs3}, one obtains the following relationship that shows how \mbox{$M_i^z(\ket{\psi}):=\expval{M^z_i}$} depend on the rotation angle around the $z$-axis:
\begin{subequations}
\begin{align}
     \label{eq:Mz1}
     M_1^z(e^{-i J_z\theta_z}\ket{\psi})&=r_z \cos (2\theta_z+\alpha_z),\\
     \label{eq:Mz2}
     M_2^z(e^{-i J_z\theta_z}\ket{\psi})&=r_z \sin (2\theta_z+\alpha_z),\\
     M_3^z(e^{-i J_z\theta_z}\ket{\psi})&=M_3^z(\ket{\psi}),
\end{align}
\end{subequations}
with
\begin{align}
     \label{eq:r_z}
     r_z&=\sqrt{(M^z_1(\ket{\psi}))^2+(M^z_2(\ket{\psi}))^2}\nonumber\\
     &=|\langle (J_x+i J_y)^2
\rangle |=|\langle (J^z_+)^2\rangle|
\end{align}
and $\alpha_z=\arccos(M_1^z/r_z)$. For rotations around the two other
axes, all the results are analogous after a cyclic exchange of labels
$x\to y\to z$.

We see that with varying the rotation angle $\theta_k$ around the $k$ axis, the first two components of $\v{M}^k(\ket{\psi})$ form a circle of a radius~$r_k$, while the third component is invariant. It is then clear that the larger the~$r_k$ is, the higher the resolution of the measurement of the rotation angle. Writing the general pure spin-1 state as
\begin{equation}
    \!\!\!\ket{\xi}\!=\! \sqrt{p^k_\uparrow} e^{i\phi^k_\uparrow} \ket{\uparrow}_k+\sqrt{1\!-\!p^k_\uparrow\!-\!p^k_\downarrow} \ket{0}_k+\sqrt{p^k_\downarrow} e^{i\phi^k_\downarrow} \ket{\downarrow}_k,\!
\end{equation}
from Eq.~\eqref{eq:r_z} we get 
\begin{equation}
    r_k=2\sqrt{p_\downarrow^k p_\uparrow^k} \leq p_\downarrow^k +p_\uparrow^k \leq \langle J_k^2 \rangle, 
\end{equation}
where we also used the inequality of arithmetic and geometric means, and the fact that we are dealing with a spin-1 system. Thus, all Pythagorean means of $r_x$, $r_y$ and $r_z$ are upper bounded by the corresponding Pythagorean means of $\braket{J_x^2}$, $\braket{J_y^2}$ and $\braket{J_z^2}$. For the latter, however, we found upper bounds of $2/3$ in Sec.~\ref{sec:variances}. Since a direct calculation shows that for a perfect quantum protractor one has $r_k=2/3$ for all $k\in\{x,y,z\}$, we conclude that these states maximise all Pythagorean means of $r_k$ and, therefore, are optimal in our experiment for distinguishing between different angle rotations along all three perpendicular axes.

\subsection{Results}

As the first step of the experimental investigations, we verified our ability to generate the perfect quantum protractor state [Eq.~\eqref{eq:exp_protractor}]. This is done by performing the full-state tomography~\cite{Kopciuch2022, Kopciuch2024optimized}, the results of which are shown in Fig.~\ref{fig:tomography_protractor}. The obtained results confirm that, using the sequence described by Eq.~\eqref{eq:generate_psi}, we generated the perfect protractor state with a fidelity of 0.995.

\begin{figure}[t]
    \centering
    \includegraphics[width = 0.9\columnwidth]{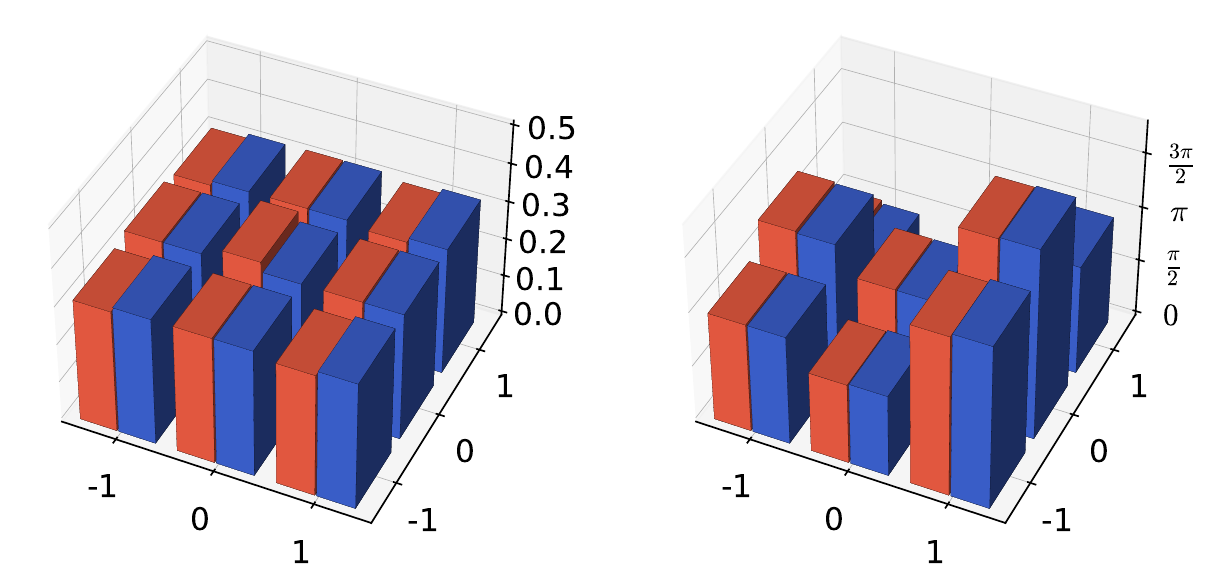}
    \caption{\label{fig:tomography}\textbf{Quantum tomography of a perfect quantum protractor.} Comparison between the reconstructed density matrix of our system (red) and the perfect protractor state (blue). The left graph shows the amplitudes of the respective matrix elements, while the right presents phases. The reconstructed data shows the fidelity of 0.995.}
\label{fig:tomography_protractor}
\end{figure}

Our main objective is to demonstrate that perfect quantum protractors form the optimal metrological resources for the considered experimental task of estimating rotation angles and to show that it saturate the upper bound for $r_k$. To achieve this, we conducted a series of measurements, in which we first applied magnetic pulses to rotate the system by angles from 0 to $\pi$ around the $z$ axis, and then measured the values of $M_1^z$ and $M_2^z$. In order to also assess the behaviour of $\v{M}^x$ and $\v{M}^y$ under rotations around $x$ and $y$, respectively, we introduced an additional $\pi/2$ rotation just before the measurement, effectively translating our measurement axis (previously aligned with $z$) into $x$ or $y$. Furthermore, we have performed this experiment not only for the initial quantum protractor state, but also for three other arbitrarily selected states of a spin-1 system. This way we can demonstrate different performance of these states for the considered metrological task (see Fig.~\ref{fig:prot_result}).

\begin{figure*}[t]
    \centering
    \includegraphics[width = 0.9\textwidth]{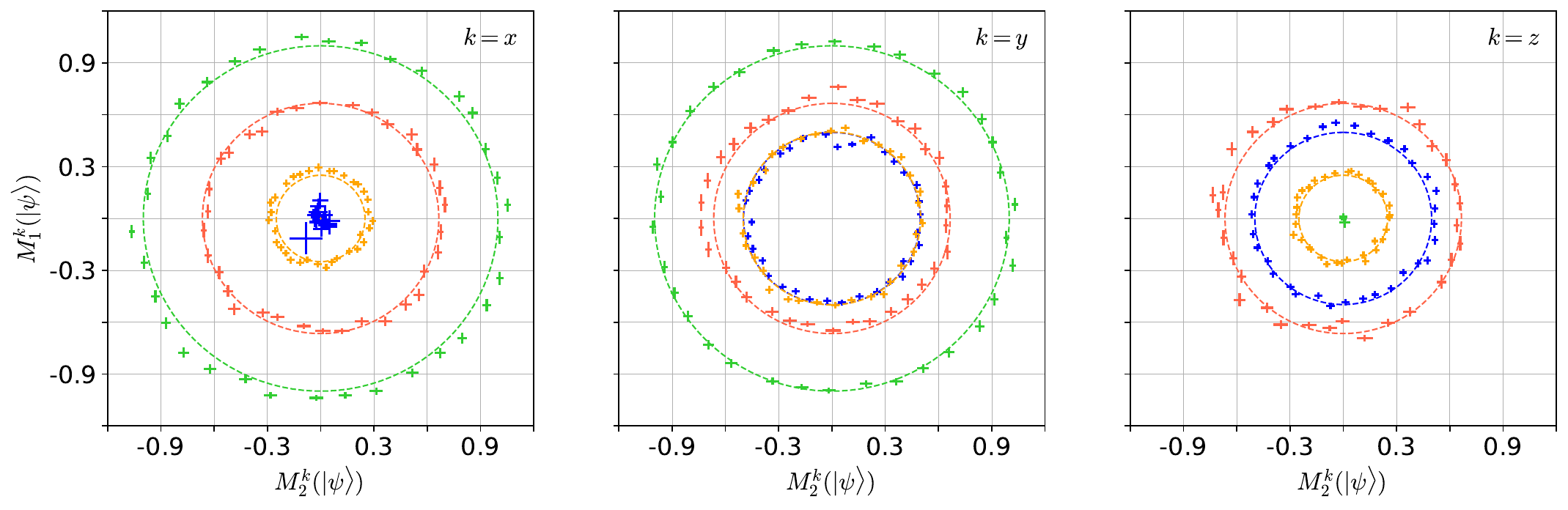}
    \caption{\textbf{Experimental estimation of rotation angles.} Results of measuring $M_1^k(\ket{\psi})$ and $M_2^k(\ket{\psi})$ for $k\in\{x,y,z\}$ (left to right) for different initial states (different colours) under rotations by $\pi l/34$ for $l\in\{0,\dots 33\}$. Points indicate the results of a single-shot experiment for a given rotation angle (error bars show propagated uncertainty of fitting to Eq.~\eqref{eq:signal}), while the dashed lines show the theoretical dependence that can be calculated using Eqs.~\eqref{eq:Mz1}-\eqref{eq:Mz2}. Colours correspond to: the perfect quantum protractor $\ket{\psi}$ -- red, $\ket{-1}_x$ -- blue, $\ket{0}_z$ -- green, and $R_y(\pi/4)\ket{-1}_x$ -- orange.
    \label{fig:prot_result}}
\end{figure*}

As can be seen in Fig.~\ref{fig:prot_result}, all experimental results match the values expected from the theoretical considerations. As predicted, for the protractor state, we observe the measurement results lying on circles of equal radius of approximately 2/3 for rotations around all three axes. We also note that, as in the case of $\ket{0}_z$ state, it is possible to achieve a higher radius (so also a higher resolution) for rotations around one or two axes, but for the price of a significant reduction for the other axis, so that according to performance measures based on Pythagorean means, perfect quantum protractors are the best states for the task.


\section{Conclusions and outlook}
\label{sec:conclusions}

In this paper we have analysed a special family of quantum states, the so-called perfect quantum protractors, that possess particular properties with respect to the rotation group. Namely, such states are simultaneously given by a uniform superposition over the eigenstates of angular momentum operators $J_x$, $J_y$ and $J_z$. As such, they maximise average uncertainty of these three observables, and they also generate three full bases under the rotations around three perpendicular axes. We also showed that these properties make perfect quantum protractors an optimal resource for a metrological task of estimating an angle of rotation around one of the \emph{a priori} unknown perpendicular axes, and demonstrated it experimentally with spin-1 systems of the atomic vapours of rubidium-87.

The optimality of perfect quantum protractors in the introduced metrological task begs for a comparison with spin anticoherent states, which have been investigated in recent years as optimal quantum states for angle measurements~\cite{goldberg2018quantum,chryssomalakos2017optimal,martin2020optimal,bouchard2017quantum,ferretti2024generating}. As already mentioned, all perfect quantum protractors are at the same time spin 1-anticoherent states, i.e., the expectation value of angular momentum operator along any axis vanishes for them. They fall short of being spin 2-anticoherent states: while the expectation value of squared angular momentum operators along $x$, $y$, and $z$ axes are maximal (and thus equal to each other), for other rotation axes they differ. Thus, while both types of states can be used for multiparameter estimation tasks, in the case of quantum protractors we cannot choose completely arbitrary rotation axes, but rather we are restricted to a distinguished set of three perpendicular ones. 

However, the biggest difference between the two types of states comes from the metrological success metric used to define them. Anticoherent states emerged from the studies of quantum states that minimize the probability that after rotation they project onto themselves, averaged uniformly over all rotation axes. One thus often focuses on sensitivity to small rotation angles and tries to maximise the change of state under rotation around any axis. Also, the nature of this metric means that in order to measure a rotation angle using spin anticoherent states, we need to perform many measurements and reconstruct the rotation angle from measurement statistics (simply because we need to estimate the overlap). On the other hand, perfect quantum protractors are defined as states that, under rotations around three perpendicular axes, transform through the largest possible number of mutually orthogonal states. As a result, they can distinguish between all rotation angles, not being limited to small rotations. Moreover, they can be used for single-shot measurements of rotation angles $\{2\pi k/(2j+1)\}_{k=0}^{2j}$, i.e., one can unambiguously distinguish between a set of such rotation angles using a single measurement, without the need of collecting measurement statistics. In Appendix~\ref{app:anticoherent}, we explicitly illustrate these differences using anticoherent and protractor states for a spin-3 system.

Since we were only able to analyse the problem of perfect quantum protractors for systems with total angular momentum $j<4$ (proving that perfect quantum protractors exist when $j\in\{1,3/2,3\}$ and numerically verifying it for $j=7/2$), the main open problem is the existence of perfect quantum protractors for all values of $j\geq 4$. It is clear that direct algebraic methods that we have used here become less and less handy for higher angular momenta, and so it seems that more general methods employing the representation theory of Lie groups could be helpful. Alternatively, smarter numerical methods could also be developed.

More abstractly, one can also generalise the studied problem to arbitrary Lie groups (e.g. SU(3)) with generators $\{g_1,\dots,g_n\}$. The question then would be to find a state $\ket{\psi}$ that transforms under each $e^{ig_kt}$ in such a way that it generates an orthogonal basis. This problem is, in a way, similar to the problem of generating symmetric informationally complete quantum measurements (SIC-POVMs)~\cite{zauner2011quantum,renes2004symmetric} via the group action on the fiducial state. Here, we are also looking for a fiducial state $\ket{\psi}$ that under the action of $n$ subgroups (instead of the whole group) generates orthogonal bases (instead of SIC-POVMs). Alternatively, one could also investigate whether quantum protractors can, under rotations, generate bases of extreme spin coherence~\cite{rudzinski2024orthonormal}.

Another mathematical perspective on the problem investigated here can be obtained by focusing on the uniform superposition property of perfect quantum protractors. Let us recall that in Refs.~\cite{korzekwa2014operational,idel2015sinkhorn} it was shown that for any two observables, $A$ and $B$, there always exist pure quantum states that are simultaneously a uniform superposition over the eigenstates of both $A$ and $B$. Here, we showed that for three particularly important observables, $J_x$, $J_y$ and $J_z$, there can also exist states that have this uniform property in three different bases. More generally, one can then investigate under what conditions on observables $A_1,\dots A_n$, it is possible to find states that are a simultaneous uniform superposition with respect to all $n$ eigenbases. For example, in prime power dimensions~$d$, where it is known that $d+1$ mutually unbiased bases exist~\cite{wootters1989optimal}, one can always find states that are a uniform superposition in $d$ different bases. Of course, these kind of results would have direct consequences for certainty relations, as discussed in this paper.

Finally, on a more applicational side, it would be interesting to investigate whether perfect quantum protractors may be useful beyond simple metrological tasks described here. For example, one could try to define an entangled generalisation of quantum protractors, in order to obtain the analogue of the phase estimation improvement from the standard quantum limit to the Heisenberg limit~\cite{giovannetti2004quantum}. In fact, the results presented in Sec.~\ref{sec:entanglement} of this paper already prove that entanglement is necessary to create the known perfect quantum protractors for $j\in\{1,3/2,3\}$ while combining multiple subsystems. The question, however, remains whether there exists some general bound quantitatively linking the amount of entanglement between subsystems and their ability to act as a perfect quantum protractor.


\section*{Acknowledgements} 

KK would like to thank Cristina C\^{i}rstoiu, Rafa{\l} Demkowicz-Dobrza{\'n}ski, Iman Marvian and Karol \.{Z}yczkowski for useful comments and discussions. KK acknowledges financial support by the Foundation for Polish Science through TEAM-NET project (contract no. POIR.04.04.00-00-17C1/18-00). This work was also supported by the National Science Centre, Poland within the Sonata Bis programme (Grant No. 2019/34/E/ST2/00440). MK would like to acknowledge the support from the Excellence Initiative – Research University of the Jagiellonian University in Kraków.


\appendix


\section{Angular momentum operators and eigenstates}
\label{app:angular}

With the matrix representation of the spin-$j$ angular-momentum operators, given in the eigenbasis of~$J_z$, i.e., when
\begin{equation}
    (J_z)_{kl}=(j+1-k)\delta_{kl},
\end{equation}
for $k,l\in\{1,\dots,2j+1\}$ and $\delta_{kl}$ denoting the Kronecker delta, we have
\begin{subequations}
\begin{align}
   \!\! (J_x)_{kl}&=\frac{1}{2}\left(\delta_{l,l+1}\!+\!\delta_{k+1,l}\right) \sqrt{(j\!+\!1)(k\!+\!l\!-\!1)\!-\!kl},\\
   \!\! (J_y)_{kl}&=\frac{i}{2}\left(\delta_{k,l+1}\!-\!\delta_{k+1,l}\right) \sqrt{(j\!+\!1)(k\!+\!l\!-\!1)\!-\!kl}.\!
\end{align}
\end{subequations}
Below, we provide several explicit matrix representations for $J_x$ and $J_y$, together with their eigenstates represented in the $J_z$ eigenbasis. To shorten the notation, we will omit the $z$ subscript for the eigenstates of $J_z$ (e.g., $\ket{\uparrow}$ should be read as $\ket{\uparrow}_z$).

\subsection*[Spin-1/2]{\texorpdfstring{$j=1/2$}{Spin-1/2}}

\begin{align}
    &J_x=\frac{1}{2}\begin{pmatrix}
    0 & 1\\
    1 & 0
    \end{pmatrix},\quad
    J_y=\frac{1}{2}\begin{pmatrix}
    0 & -i\\
    i & 0
    \end{pmatrix},
\end{align}

\begin{subequations}
\begin{align}
    &\ket{\uparrow}_x=\frac{\ket{\uparrow}+\ket{\downarrow}}{\sqrt{2}},\quad
    \ket{\downarrow}_x=\frac{\ket{\uparrow}-\ket{\downarrow}}{\sqrt{2}},\\
    &\ket{\uparrow}_y=\frac{\ket{\uparrow}+i\ket{\downarrow}}{\sqrt{2}},\quad
    \ket{\downarrow}_y=\frac{\ket{\uparrow}-i\ket{\downarrow}}{\sqrt{2}}.
\end{align}
\end{subequations}

\subsection*[Spin-1]{\texorpdfstring{$j=1$}{Spin-1}}

\begin{align}
    &J_x=\frac{1}{\sqrt{2}}\begingroup
\setlength\arraycolsep{3pt}\begin{pmatrix}
    0 & 1 & 0\\
    1 & 0 & 1\\
    0 & 1 & 0
    \end{pmatrix},\quad
    J_y=\frac{i}{\sqrt{2}}\begin{pmatrix}
    0 & -1 & 0\\
    1 & 0 & -1\\
    0 & 1 & 0
    \end{pmatrix},
    \endgroup
\end{align}

\begin{subequations}
\begin{align}
    \ket{\uparrow}_x&=\frac{1}{2}(\ket{\uparrow}+\sqrt{2}\ket{0}+\ket{\downarrow}),\\
    \ket{0}_x&=\frac{1}{\sqrt{2}}(-\ket{\uparrow}+\ket{\downarrow}),\\
    \ket{\downarrow}_x&=\frac{1}{2}(\ket{\uparrow}-\sqrt{2}\ket{0}+\ket{\downarrow}),
\end{align}
\end{subequations}

\begin{subequations}
\begin{align}
    \ket{\uparrow}_y&=\frac{1}{2}(\ket{\uparrow}+i\sqrt{2}\ket{0}-\ket{\downarrow}),\\
    \ket{0}_y&=\frac{1}{\sqrt{2}}(\ket{\uparrow}+\ket{\downarrow}),\\
    \ket{\downarrow}_y&=\frac{1}{2}(-\ket{\uparrow}+i\sqrt{2}\ket{0}+\ket{\downarrow}).
\end{align}
\end{subequations}

\subsection*[Spin-3/2]{\texorpdfstring{$j=3/2$}{Spin-3/2}}

\begin{align}
    J_x=\frac{1}{2}&\begin{pmatrix}
    0 & \sqrt{3} & 0 & 0\\
    \sqrt{3} & 0 & 2 & 0\\
    0 & 2 & 0 & \sqrt{3}\\
    0 & 0 & \sqrt{3} & 0
    \end{pmatrix},\\[1em]
    J_y=\frac{i}{2}&\begin{pmatrix}
    0 & -\sqrt{3} & 0 & 0\\
    \sqrt{3} & 0 & -2 & 0\\
    0 & 2 & 0 & -\sqrt{3}\\
    0 & 0 & \sqrt{3} & 0
    \end{pmatrix},
\end{align}

\begin{subequations}
    \begin{align}
        \ket{\Uparrow}_x=&\frac{1}{2\sqrt{2}}\left(\ket{\Uparrow}+\sqrt{3}\ket{\uparrow}+\sqrt{3}\ket{\downarrow}+\ket{\Downarrow}\right),\\
        \ket{\uparrow}_x=&\frac{1}{2\sqrt{2}}\left(-\sqrt{3}\ket{\Uparrow}-\ket{\uparrow}+\ket{\downarrow}+\sqrt{3}\ket{\Downarrow}\right),\\
        \ket{\downarrow}_x=&\frac{1}{2\sqrt{2}}\left(\sqrt{3}\ket{\Uparrow}-\ket{\uparrow}-\ket{\downarrow}+\sqrt{3}\ket{\Downarrow}\right),\\
        \ket{\Downarrow}_x=&\frac{1}{2\sqrt{2}}\left(-\ket{\Uparrow}+\sqrt{3}\ket{\uparrow}-\sqrt{3}\ket{\downarrow}+\ket{\Downarrow}\right),
    \end{align}
\end{subequations}

\begin{subequations}
    \begin{align}
        \ket{\Uparrow}_y=&\frac{1}{2\sqrt{2}}\left(\ket{\Uparrow}+i\sqrt{3}\ket{\uparrow}-\sqrt{3}\ket{\downarrow}-i\ket{\Downarrow}\right),\\
        \ket{\uparrow}_y=&\frac{1}{2\sqrt{2}}\left(i\sqrt{3}\ket{\Uparrow}-\ket{\uparrow}+i\ket{\downarrow}-\sqrt{3}\ket{\Downarrow}\right),\\
        \ket{\downarrow}_y=&\frac{1}{2\sqrt{2}}\left(-\sqrt{3}\ket{\Uparrow}+i\ket{\uparrow}-\ket{\downarrow}+i\sqrt{3}\ket{\Downarrow}\right),\\
        \ket{\Downarrow}_y=&\frac{1}{2\sqrt{2}}\left(-i\ket{\Uparrow}-\sqrt{3}\ket{\uparrow}+i\sqrt{3}\ket{\downarrow}+\ket{\Downarrow}\right).
    \end{align}
\end{subequations}

\subsection*[Spin-2]{\texorpdfstring{$j=2$}{Spin-2}}
\begingroup\setlength\arraycolsep{3pt}
\begin{align}
    J_x=&\begin{pmatrix}
    0 & 1 & 0 & 0 & 0\\
    1 & 0 &\sqrt{\frac{3}{2}} & 0 &0\\
    0 & \sqrt{\frac{3}{2}} & 0 & \sqrt{\frac{3}{2}} & 0\\
    0 & 0 &\sqrt{\frac{3}{2}} & 0 & 1\\
    0&0&0&1&0
    \end{pmatrix},\\[1em]
    J_y=i&\begin{pmatrix}
    0 & -1 & 0 & 0 & 0\\
    1 & 0 &-\sqrt{\frac{3}{2}} & 0 &0\\
    0 & \sqrt{\frac{3}{2}} & 0 & -\sqrt{\frac{3}{2}} & 0\\
    0 & 0 &\sqrt{\frac{3}{2}} & 0 & -1\\
    0&0&0&1&0
    \end{pmatrix},
\end{align}
\endgroup
\begin{subequations}
    \begin{align}
        \ket{\Uparrow}_x=&\frac{1}{4}\left(\ket{\Uparrow}+2\ket{\uparrow}+\sqrt{6}\ket{0}+2\ket{\downarrow}+\ket{\Downarrow}\right),\\
        \ket{\uparrow}_x=&\frac{1}{2}\left(-\ket{\Uparrow}+\ket{\uparrow}-\ket{\downarrow}+\ket{\Downarrow}\right),\\
        \ket{0}_x=&\frac{1}{2}\left(\sqrt{\frac{3}{2}}\ket{\Uparrow}-\ket{0}+\sqrt{\frac{3}{2}}\ket{\Downarrow}\right),\\
        \ket{\downarrow}_x=&\frac{1}{2}\left(-\ket{\Uparrow}-\ket{\uparrow}+\ket{\downarrow}+\ket{\Downarrow}\right),\\
        \ket{\Downarrow}_x=&\frac{1}{4}\left(\ket{\Uparrow}-2\ket{\uparrow}+\sqrt{6}\ket{0}-2\ket{\downarrow}+\ket{\Downarrow}\right),
    \end{align}
\end{subequations}
\vspace{-20pt}
\begin{subequations}
    \begin{align}
        \ket{\Uparrow}_y=&\frac{1}{4}\left(\ket{\Uparrow}+2i\ket{\uparrow}-\sqrt{6}\ket{0}-2i\ket{\downarrow}+\ket{\Downarrow}\right),\\
        \ket{\uparrow}_y=&\frac{1}{2}\left(-\ket{\Uparrow}-i\ket{\uparrow}-i\ket{\downarrow}+\ket{\Downarrow}\right),\\
        \ket{0}_y=&\frac{1}{2}\left(\sqrt{\frac{3}{2}}\ket{\Uparrow}+\ket{0}+\sqrt{\frac{3}{2}}\ket{\Downarrow}\right),\\
        \ket{\downarrow}_y=&\frac{1}{2}\left(-\ket{\Uparrow}+i\ket{\uparrow}+i\ket{\downarrow}+\ket{\Downarrow}\right),\\
        \ket{\Downarrow}_y=&\frac{1}{4}\left(\ket{\Uparrow}-2i\ket{\uparrow}-\sqrt{6}\ket{0}+2i\ket{\downarrow}+\ket{\Downarrow}\right).
    \end{align}
\end{subequations}

\subsection*[Spin-5/2]{\texorpdfstring{$j=5/2$}{Spin-5/2}}
\begin{align}
    J_x=\frac{1}{2}&\begingroup
\setlength\arraycolsep{3pt}\begin{pmatrix}
    0 & \sqrt{5} & 0 & 0 & 0 & 0\\
    \sqrt{5} & 0 & 2\sqrt{2} & 0 & 0 & 0\\
    0 & 2\sqrt{2} & 0 & 3 & 0 & 0\\
    0 & 0 & 3 & 0 & 2\sqrt{2} & 0\\
    0 & 0 & 0 & 2\sqrt{2} & 0 & \sqrt{5}\\
    0 & 0 & 0 & 0 & \sqrt{5} & 0
    \end{pmatrix},\endgroup\\[1em]
    J_y=\frac{i}{2}&\begingroup
\setlength\arraycolsep{3pt}\begin{pmatrix}
   0 & -\sqrt{5} & 0 & 0 & 0 & 0\\
    \sqrt{5} & 0 & -2\sqrt{2} & 0 & 0 & 0\\
    0 & 2\sqrt{2} & 0 & -3 & 0 & 0\\
    0 & 0 & 3 & 0 & -2\sqrt{2} & 0\\
    0 & 0 & 0 & 2\sqrt{2} & 0 & -\sqrt{5}\\
    0 & 0 & 0 & 0 & \sqrt{5} & 0
    \end{pmatrix},
    \endgroup
\end{align}

\begin{subequations}
    \begin{align}
        \ket{\Uuparrow}_x=&\frac{1}{4\sqrt{2}}\left(\ket{\Uuparrow}+\sqrt{5}\ket{\Uparrow}+\sqrt{10}\ket{\uparrow}\right.\\ \nonumber&\quad\quad\left.+\sqrt{10}\ket{\downarrow}+\sqrt{5}\ket{\Downarrow}+\ket{\Ddownarrow}\right),\\
        \ket{\Uparrow}_x=&\frac{1}{4\sqrt{2}}\left(-\sqrt{5}\ket{\Uuparrow} -3 \ket{\Uparrow}-\sqrt{2}\ket{\uparrow}\right.\\ \nonumber&\quad\quad\left. +\sqrt{2}\ket{\downarrow}+ 3\ket{\Downarrow}+ \sqrt{5}\ket{\Ddownarrow}\right),\\
        \ket{\uparrow}_x=&\frac{1}{4}\left(\sqrt{5}\ket{\Uuparrow}+\ket{\Uparrow}-\sqrt{2}\ket{\uparrow}\right.\\ \nonumber&\quad\left.-\sqrt{2}\ket{\downarrow}+\ket{\Downarrow}+\sqrt{5}\ket{\Ddownarrow}\right),
        \end{align}
        \vspace{9pt}
        \begin{align}
        \ket{\downarrow}_x=&\frac{1}{4}\left(-\sqrt{5}\ket{\Uuparrow}+\ket{\Uparrow}+\sqrt{2}\ket{\uparrow}\right.\\ \nonumber&\quad\left.-\sqrt{2}\ket{\downarrow}-\ket{\Downarrow}+\sqrt{5}\ket{\Ddownarrow}\right),\\
        \ket{\Downarrow}_x=&\frac{1}{4\sqrt{2}}\left(\sqrt{5}\ket{\Uuparrow} -3 \ket{\Uparrow}+\sqrt{2}\ket{\uparrow}\right.\\ \nonumber&\quad\quad\left. +\sqrt{2}\ket{\downarrow}- 3\ket{\Downarrow}+ \sqrt{5}\ket{\Ddownarrow}\right)\\
        \ket{\Ddownarrow}_x=&\frac{1}{4\sqrt{2}}\left(-\ket{\Uuparrow}+\sqrt{5}\ket{\Uparrow}-\sqrt{10}\ket{\uparrow}\right.\\ \nonumber&\quad\quad\left.+\sqrt{10}\ket{\downarrow}-\sqrt{5}\ket{\Downarrow}+\ket{\Ddownarrow}\right),
    \end{align}
\end{subequations}

\begin{subequations}
    \begin{align}
        \ket{\Uuparrow}_y=&\frac{1}{4\sqrt{2}}\left(-i\ket{\Uuparrow}+\sqrt{5}\ket{\Uparrow}+i\sqrt{10}\ket{\uparrow}\right.\\ \nonumber&\quad\quad\left.-\sqrt{10}\ket{\downarrow}-i\sqrt{5}\ket{\Downarrow}+\ket{\Ddownarrow}\right),\\
        \ket{\Uparrow}_y=&\frac{1}{4\sqrt{2}}\left(i\sqrt{5}\ket{\Uuparrow} -3 \ket{\Uparrow}-i\sqrt{2}\ket{\uparrow}\right.\\ \nonumber&\quad\quad\left. -\sqrt{2}\ket{\downarrow}- 3i\ket{\Downarrow}+ \sqrt{5}\ket{\Ddownarrow}\right),\\
        \ket{\uparrow}_y=&\frac{1}{4}\left(-i\sqrt{5}\ket{\Uuparrow}+\ket{\Uparrow}-i\sqrt{2}\ket{\uparrow}\right.\\ \nonumber&\quad\left.+\sqrt{2}\ket{\downarrow}-i\ket{\Downarrow}+\sqrt{5}\ket{\Ddownarrow}\right),\\
        \ket{\downarrow}_y=&\frac{1}{4}\left(i\sqrt{5}\ket{\Uuparrow}+\ket{\Uparrow}+i\sqrt{2}\ket{\uparrow}\right.\\ \nonumber&\quad\left.+\sqrt{2}\ket{\downarrow}+i\ket{\Downarrow}+\sqrt{5}\ket{\Ddownarrow}\right),\\
        \ket{\Downarrow}_y=&\frac{1}{4\sqrt{2}}\left(-i\sqrt{5}\ket{\Uuparrow} -3 \ket{\Uparrow}+i\sqrt{2}\ket{\uparrow}\right.\\ \nonumber&\quad\quad\left. -\sqrt{2}\ket{\downarrow}+3i\ket{\Downarrow}+ \sqrt{5}\ket{\Ddownarrow}\right),\\
        \ket{\Ddownarrow}_y=&\frac{1}{4\sqrt{2}}\left(i\ket{\Uuparrow}+\sqrt{5}\ket{\Uparrow}-i\sqrt{10}\ket{\uparrow}\right.\\ \nonumber&\quad\quad\left.-\sqrt{10}\ket{\downarrow}+i\sqrt{5}\ket{\Downarrow}+\ket{\Ddownarrow}\right).
    \end{align}
\end{subequations}

\subsection*[Spin-3]{\texorpdfstring{$j=3$}{Spin-3}}

\begingroup\setlength\arraycolsep{2pt}
\begin{align}
    J_x=\frac{1}{2}&\begin{pmatrix}
    0 & \sqrt{6} & 0 & 0 & 0 & 0 & 0\\
    \sqrt{6} & 0 & \sqrt{10} & 0 & 0 & 0 & 0\\
    0 & \sqrt{10} & 0 & 2\sqrt{3} & 0 & 0 & 0\\
    0 & 0 & 2\sqrt{3} & 0 & 2\sqrt{3} & 0& 0\\
    0 & 0 & 0 & 2\sqrt{3} & 0 & \sqrt{10}& 0\\
    0 & 0 & 0 & 0 & \sqrt{10} & 0 & \sqrt{6}\\
    0& 0& 0& 0& 0 & \sqrt{6} & 0
    \end{pmatrix},\\[0.93em]
    J_y=\frac{i}{2}&\begin{pmatrix}
  0 & -\sqrt{6} & 0 & 0 & 0 & 0 & 0\\
    \sqrt{6} & 0 & -\sqrt{10} & 0 & 0 & 0 & 0\\
    0 & \sqrt{10} & 0 & -2\sqrt{3} & 0 & 0 & 0\\
    0 & 0 & 2\sqrt{3} & 0 & -2\sqrt{3} & 0& 0\\
    0 & 0 & 0 & 2\sqrt{3} & 0 & -\sqrt{10}& 0\\
    0 & 0 & 0 & 0 & \sqrt{10} & 0 & -\sqrt{6}\\
    0& 0& 0& 0& 0 & \sqrt{6} & 0
    \end{pmatrix},
\end{align}
\endgroup
\vspace{-15pt}
\begin{subequations}
    \begin{align}
        \ket{\Uuparrow}_x=&\frac{1}{8}\left(\ket{\Uuparrow}+\sqrt{6}\ket{\Uparrow}+\sqrt{15}\ket{\uparrow}+2\sqrt{5}\ket{0}\right.\\ \nonumber&\quad\left.+\sqrt{15}\ket{\downarrow}+\sqrt{6}\ket{\Downarrow}+\ket{\Ddownarrow}\right),\\
        \ket{\Uparrow}_x=&\frac{1}{4}\left(-\sqrt{\frac{3}{2}}\ket{\Uuparrow} -2 \ket{\Uparrow}-\sqrt{\frac{5}{2}}\ket{\uparrow}\right.\\ \nonumber&\quad\left. +\sqrt{\frac{5}{2}}\ket{\downarrow}+ 2\ket{\Downarrow}+ \sqrt{\frac{3}{2}}\ket{\Ddownarrow}\right),\\
        \ket{\uparrow}_x=&\frac{1}{8}\left(\sqrt{15}\ket{\Uuparrow}+\sqrt{10}\ket{\Uparrow}-\ket{\uparrow}-2\sqrt{3}\ket{0}\right.\\ \nonumber&\quad\left.-\ket{\downarrow}+\sqrt{10}\ket{\Downarrow}+\sqrt{15}\ket{\Ddownarrow}\right),\\
        \ket{0}_x=&\frac{1}{4}\left(-\sqrt{5}\ket{\Uuparrow}+ \sqrt{3}\ket{\uparrow} -\sqrt{3}\ket{\downarrow}+ \sqrt{5}\ket{\Ddownarrow}\right),\\
        \ket{\downarrow}_x=&\frac{1}{8}\left(\sqrt{15}\ket{\Uuparrow}-\sqrt{10}\ket{\Uparrow}-\ket{\uparrow}+2\sqrt{3}\ket{0}\right.\\ \nonumber&\quad\left.-\ket{\downarrow}-\sqrt{10}\ket{\Downarrow}+\sqrt{15}\ket{\Ddownarrow}\right),\\
        \ket{\Downarrow}_x=&\frac{1}{4}\left(-\sqrt{\frac{3}{2}}\ket{\Uuparrow} +2 \ket{\Uparrow}-\sqrt{\frac{5}{2}}\ket{\uparrow}\right.\\ \nonumber&\quad\left. +\sqrt{\frac{5}{2}}\ket{\downarrow}- 2\ket{\Downarrow}+ \sqrt{\frac{3}{2}}\ket{\Ddownarrow}\right),\\
        \ket{\Ddownarrow}_x=&\frac{1}{8}\left(\ket{\Uuparrow}-\sqrt{6}\ket{\Uparrow}+\sqrt{15}\ket{\uparrow}-2\sqrt{5}\ket{0}\right.\\ \nonumber&\quad\left.+\sqrt{15}\ket{\downarrow}-\sqrt{6}\ket{\Downarrow}+\ket{\Ddownarrow}\right),
    \end{align}
\end{subequations}

\begin{subequations}
    \begin{align}
        \ket{\Uuparrow}_y=&\frac{1}{8}\left(-\ket{\Uuparrow}-i\sqrt{6}\ket{\Uparrow}+\sqrt{15}\ket{\uparrow}+2i\sqrt{5}\ket{0}\right.\\ \nonumber&\quad\left.-\sqrt{15}\ket{\downarrow}-i\sqrt{6}\ket{\Downarrow}+\ket{\Ddownarrow}\right),\\
        \ket{\Uparrow}_y=&\frac{1}{4}\left(\sqrt{\frac{3}{2}}\ket{\Uuparrow} +2i \ket{\Uparrow}-\sqrt{\frac{5}{2}}\ket{\uparrow}\right.\\ \nonumber&\quad\left. -\sqrt{\frac{5}{2}}\ket{\downarrow}- 2i\ket{\Downarrow}+ \sqrt{\frac{3}{2}}\ket{\Ddownarrow}\right),\\
        \ket{\uparrow}_y=&\frac{1}{8}\left(-\sqrt{15}\ket{\Uuparrow}-i\sqrt{10}\ket{\Uparrow}-\ket{\uparrow}-2i\sqrt{3}\ket{0}\right.\\ \nonumber&\quad\left.+\ket{\downarrow}-i\sqrt{10}\ket{\Downarrow}+\sqrt{15}\ket{\Ddownarrow}\right),\\
        \ket{0}_y=&\frac{1}{4}\left(\sqrt{5}\ket{\Uuparrow}+ \sqrt{3}\ket{\uparrow} +\sqrt{3}\ket{\downarrow}+ \sqrt{5}\ket{\Ddownarrow}\right),\\
        \ket{\downarrow}_y=&\frac{1}{8}\left(-\sqrt{15}\ket{\Uuparrow}+i\sqrt{10}\ket{\Uparrow}-\ket{\uparrow}+2i\sqrt{3}\ket{0}\right.\\ \nonumber&\quad\left.+\ket{\downarrow}+i\sqrt{10}\ket{\Downarrow}+\sqrt{15}\ket{\Ddownarrow}\right),\\
        \ket{\Downarrow}_y=&\frac{1}{4}\left(\sqrt{\frac{3}{2}}\ket{\Uuparrow} -2i \ket{\Uparrow}-\sqrt{\frac{5}{2}}\ket{\uparrow}\right.\\ \nonumber&\quad\left. -\sqrt{\frac{5}{2}}\ket{\downarrow}+2i\ket{\Downarrow}+ \sqrt{\frac{3}{2}}\ket{\Ddownarrow}\right),\\
        \ket{\Ddownarrow}_y=&\frac{1}{8}\left(-\ket{\Uuparrow}+i\sqrt{6}\ket{\Uparrow}+\sqrt{15}\ket{\uparrow}-2i\sqrt{5}\ket{0}\right.\\ \nonumber&\quad\left.-\sqrt{15}\ket{\downarrow}+i\sqrt{6}\ket{\Downarrow}+\ket{\Ddownarrow}\right).
    \end{align}
\end{subequations}

\subsection*[Spin-7/2]{\texorpdfstring{$j=7/2$}{Spin-7/2}}
\vspace{-5pt}
\begingroup\setlength\arraycolsep{2pt}
\begin{align}  
    \!J_x\!=\!\frac{1}{2}&\begin{pmatrix}
    0 & \sqrt{7} & 0 & 0 & 0 & 0 & 0 & 0\\
    \sqrt{7} & 0 & 2\sqrt{3} & 0 & 0 & 0 & 0 & 0\\
    0 & 2\sqrt{3} & 0 & \sqrt{15} & 0 & 0 & 0 & 0\\
    0 & 0 & \sqrt{15} & 0 & 4 & 0 & 0 & 0\\
    0 & 0 & 0 & 4 & 0 & \sqrt{15} & 0 & 0\\
    0 & 0 & 0 & 0 & \sqrt{15} & 0 & 2\sqrt{3} & 0\\
    0 & 0 & 0 & 0 & 0 & 2\sqrt{3} & 0 &  \sqrt{7}\\
    0 & 0 & 0 & 0 & 0 & 0 & \sqrt{7} & 0
    \end{pmatrix}\!,\!\\
    \!J_y\!=\!\frac{i}{2}&\begin{pmatrix}
    0 & -\sqrt{7} & 0 & 0 & 0 & 0 & 0 & 0\\
    \sqrt{7} & 0 & -2\sqrt{3} & 0 & 0 & 0 & 0 & 0\\
    0 & 2\sqrt{3} & 0 & -\sqrt{15} & 0 & 0 & 0 & 0\\
    0 & 0 & \sqrt{15} & 0 & -4 & 0 & 0 & 0\\
    0 & 0 & 0 & 4 & 0 & -\sqrt{15} & 0 & 0\\
    0 & 0 & 0 & 0 & \sqrt{15} & 0 & -2\sqrt{3} & 0\\
    0 & 0 & 0 & 0 & 0 & 2\sqrt{3} & 0 &  -\sqrt{7}\\
    0 & 0 & 0 & 0 & 0 & 0 & \sqrt{7} & 0
    \end{pmatrix}\!,\!
\end{align}
\endgroup
\vspace{-23pt}
\begin{subequations}
    \begin{align}
        \ket{\Uuparrow}_x=&\frac{1}{8\sqrt{2}}\left(\ket{\Uuparrow}+\sqrt{7}\ket{\Uparrow}+\sqrt{21}\ket{\twoheaduparrow}+\sqrt{35}\ket{\uparrow}\right.\\ \nonumber&\quad\quad\left.+\sqrt{35}\ket{\downarrow}+\sqrt{21}\ket{\twoheaddownarrow}+\sqrt{7}\ket{\Downarrow}+\ket{\Ddownarrow}\right),\\
        \ket{\Uparrow}_x=&\frac{1}{8\sqrt{2}}\left(-\sqrt{7}\ket{\Uuparrow} -5\ket{\Uparrow} -3 \sqrt{3}\ket{\twoheaduparrow}\right.\\ \nonumber&\quad\quad\left. -\sqrt{5}\ket{\uparrow}+ \sqrt{5}\ket{\downarrow}+ 3 \sqrt{3}\ket{\twoheaddownarrow}\right.\\ \nonumber&\quad\quad\left.+ 5\ket{\Uparrow}+ \sqrt{7}\ket{\Ddownarrow}\right),\\
        \ket{\twoheaduparrow}_x=&\frac{1}{8\sqrt{2}}\left(\sqrt{21}\ket{\Uuparrow}+ 3 \sqrt{3}\ket{\Uparrow}+ \ket{\twoheaduparrow} -\sqrt{15}\ket{\uparrow}\right.\\ \nonumber&\quad\quad\left. -\sqrt{15}\ket{\downarrow} +\ket{\twoheaddownarrow}+ 3 \sqrt{3}\ket{\Downarrow}+\sqrt{21}\ket{\Ddownarrow}\right),\\
        \ket{\uparrow}_x=&\frac{1}{8\sqrt{2}}\left(-\sqrt{35}\ket{\Uuparrow} -\sqrt{5}\ket{\Uparrow}+ \sqrt{15}\ket{\twoheaduparrow} \right.\\ \nonumber&\quad\quad\left.+3\ket{\uparrow}-3\ket{\downarrow} -\sqrt{15}\ket{\twoheaddownarrow}\right.\\ \nonumber&\quad\quad\left.+ \sqrt{5}\ket{\Downarrow} +\sqrt{35}\ket{\Ddownarrow}\right),\\
        \ket{\downarrow}_x=&\frac{1}{8\sqrt{2}}\left(\sqrt{35}\ket{\Uuparrow} -\sqrt{5}\ket{\Uparrow}- \sqrt{15}\ket{\twoheaduparrow} +3\ket{\uparrow}\right.\\ \nonumber&\quad\quad\left.+3\ket{\downarrow} -\sqrt{15}\ket{\twoheaddownarrow}-\sqrt{5}\ket{\Downarrow} +\sqrt{35}\ket{\Ddownarrow}\right),\\
        \ket{\twoheaddownarrow}_x=&\frac{1}{8\sqrt{2}}\left(-\sqrt{21}\ket{\Uuparrow}+ 3 \sqrt{3}\ket{\Uparrow}- \ket{\twoheaduparrow} \right.\\ \nonumber&\quad\quad\left.-\sqrt{15}\ket{\uparrow} +\sqrt{15}\ket{\downarrow} +\ket{\twoheaddownarrow}\right.\\ \nonumber&\quad\quad\left.- 3 \sqrt{3}\ket{\Downarrow}+\sqrt{21}\ket{\Ddownarrow}\right),\\
        \ket{\Downarrow}_x=&\frac{1}{8\sqrt{2}}\left(\sqrt{7}\ket{\Uuparrow} -5\ket{\Uparrow} +3 \sqrt{3}\ket{\twoheaduparrow} -\sqrt{5}\ket{\uparrow}\right.\\ \nonumber&\quad\quad\left.- \sqrt{5}\ket{\downarrow}+ 3 \sqrt{3}\ket{\twoheaddownarrow}- 5\ket{\Uparrow}+\sqrt{7} \ket{\Ddownarrow}\right),\\
        \ket{\Ddownarrow}_x=&\frac{1}{8\sqrt{2}}\left(-\ket{\Uuparrow}\!+\!\sqrt{7}\ket{\Uparrow}\!-\!\sqrt{21}\ket{\twoheaduparrow}\!+\!\sqrt{35}\ket{\uparrow}\right.\\ \nonumber&\quad\quad\left.-\sqrt{35}\ket{\downarrow}+\sqrt{21}\ket{\twoheaddownarrow}-\sqrt{7}\ket{\Downarrow}+\ket{\Ddownarrow}\right),
    \end{align}
\end{subequations}
\begin{subequations}
    \begin{align}
        \ket{\Uuparrow}_y=&\frac{1}{8\sqrt{2}}\left(i\ket{\Uuparrow}-\sqrt{7}\ket{\Uparrow}-i \sqrt{21}\ket{\twoheaduparrow}\right.\\ \nonumber&\quad\quad\left.+\sqrt{35}\ket{\uparrow}+i\sqrt{35}\ket{\downarrow}-\sqrt{21}\ket{\twoheaddownarrow}\right.\\ \nonumber&\quad\quad\left.-i\sqrt{7}\ket{\Downarrow}+\ket{\Ddownarrow}\right),\\
        \ket{\Uparrow}_y=&\frac{1}{8\sqrt{2}}\left(-i\sqrt{7}\ket{\Uuparrow} +5\ket{\Uparrow} +3i \sqrt{3}\ket{\twoheaduparrow}\right.\\ \nonumber&\quad\quad\left. -\sqrt{5}\ket{\uparrow}+ \sqrt{5}i\ket{\downarrow}- 3 \sqrt{3}\ket{\twoheaddownarrow}\right.\\ \nonumber&\quad\quad\left.- 5i\ket{\Uparrow}+ \sqrt{7}\ket{\Ddownarrow}\right),\\
        \ket{\twoheaduparrow}_y=&\frac{1}{8\sqrt{2}}\left(\sqrt{21}\ket{\Uuparrow}- 3 \sqrt{3}\ket{\Uparrow}-i \ket{\twoheaduparrow} \right.\\ \nonumber&\quad\quad\left.-\sqrt{15}\ket{\uparrow} -i\sqrt{15}\ket{\downarrow} -\ket{\twoheaddownarrow}\right.\\ \nonumber&\quad\quad\left.- 3i \sqrt{3}\ket{\Downarrow}+\sqrt{21}\ket{\Ddownarrow}\right),\\
        \ket{\uparrow}_y=&\frac{1}{8\sqrt{2}}\left(-i\sqrt{35}\ket{\Uuparrow} +\sqrt{5}\ket{\Uparrow}-i \sqrt{15}\ket{\twoheaduparrow}\right.\\ \nonumber&\quad\quad\left. +3\ket{\uparrow}-3i\ket{\downarrow} +\sqrt{15}\ket{\twoheaddownarrow}\right.\\ \nonumber&\quad\quad\left.-i \sqrt{5}\ket{\Downarrow} +\sqrt{35}\ket{\Ddownarrow}\right),\\
        \ket{\downarrow}_y=&\frac{1}{8\sqrt{2}}\left(i\sqrt{35}\ket{\Uuparrow} +\sqrt{5}\ket{\Uparrow}+i \sqrt{15}\ket{\twoheaduparrow} \right.\\ \nonumber&\quad\quad\left.+3\ket{\uparrow}+3i\ket{\downarrow} +\sqrt{15}\ket{\twoheaddownarrow}\right.\\ \nonumber&\quad\quad\left.+i\sqrt{5}\ket{\Downarrow} +\sqrt{35}\ket{\Ddownarrow}\right),\\
        \ket{\twoheaddownarrow}_y=&\frac{1}{8\sqrt{2}}\left(-i\sqrt{21}\ket{\Uuparrow}- 3 \sqrt{3}\ket{\Uparrow}+i \ket{\twoheaduparrow} \right.\\ \nonumber&\quad\quad\left.-\sqrt{15}\ket{\uparrow} +i\sqrt{15}\ket{\downarrow} -\ket{\twoheaddownarrow}\right.\\ \nonumber&\quad\quad\left.+ 3i \sqrt{3}\ket{\Downarrow}+\sqrt{21}\ket{\Ddownarrow}\right),\\
        \ket{\Downarrow}_y=&\frac{1}{8\sqrt{2}}\left(i\sqrt{7}\ket{\Uuparrow} +5\ket{\Uparrow} -3i \sqrt{3}\ket{\twoheaduparrow}\right.\\ \nonumber&\quad\quad\left. -\sqrt{5}\ket{\uparrow}- i\sqrt{5}\ket{\downarrow}- 3 \sqrt{3}\ket{\twoheaddownarrow}\right.\\ \nonumber&\quad\quad\left.+ 5i\ket{\Uparrow}+\sqrt{7} \ket{\Ddownarrow}\right),\\
        \ket{\Ddownarrow}_y=&\frac{1}{8\sqrt{2}}\left(-i\ket{\Uuparrow}-\sqrt{7}\ket{\Uparrow}+i\sqrt{21}\ket{\twoheaduparrow}\right.\\ \nonumber&\quad\quad\left.+\sqrt{35}\ket{\uparrow}-i\sqrt{35}\ket{\downarrow}-\sqrt{21}\ket{\twoheaddownarrow}\right.\\ \nonumber&\quad\quad\left.+i\sqrt{7}\ket{\Downarrow}+\ket{\Ddownarrow}\right).
    \end{align}
\end{subequations}


\section{Numerical phases for spin-7/2}
\label{app:numerics}

Consider the following $\ket{\psi}\in\H_{7/2}$:
\begin{equation}
    \ket{\psi} = \frac{1}{\sqrt{8}}\sum_{m=-7/2}^{7/2} e^{i\phi_m} \ket{j,m}_z,    
\end{equation}
with
\begin{subequations}
    \begin{align}    
        \phi_{7/2} &= 0,\\
        \phi_{5/2} &= \phi_{-3/2}-\phi_{3/2}+\phi_{-5/2} +\pi,\\
        \phi_{1/2} & = \theta+\frac{\phi_{3/2}+\phi_{-3/2}+\phi_{-7/2}}{2},\\
        \phi_{-1/2} & = \theta+\frac{\phi_{3/2}+\phi_{-3/2}-\phi_{-7/2}}{2}+ \pi,
    \end{align}   
\end{subequations}
where
\begin{align}
    \theta=& -\arcsin\left(\sqrt{\frac{7}{15}}\sin\left(\phi_{-5/2}+\frac{\phi_{-3/2}-\phi_{3/2}-\phi_{-7/2}}{2}\right)\right.\nonumber\\
    &\qquad\qquad~~+\frac{\phi_{-3/2}+\phi_{3/2}+\phi_{-7/2}}{2}\Bigg).
\end{align}
Then, around the following numerical values:
\begin{subequations}
    \begin{align}
        \phi_{3/2} =& -2.30181413,\\
        \phi_{-3/2} =& -0.60467766,\\
        \phi_{-5/2} =& -1.46074545,\\
        \phi_{-7/2} =& -0.57982923,
    \end{align}
\end{subequations}
the state $\ket{\psi}$, up to numerical precision, becomes a perfect quantum protractor.


\section{Comparison with spin 2-anticoherent state for spin-3 system}
\label{app:anticoherent}

\begin{figure}[t]
    \centering
    \includegraphics[width=\columnwidth]{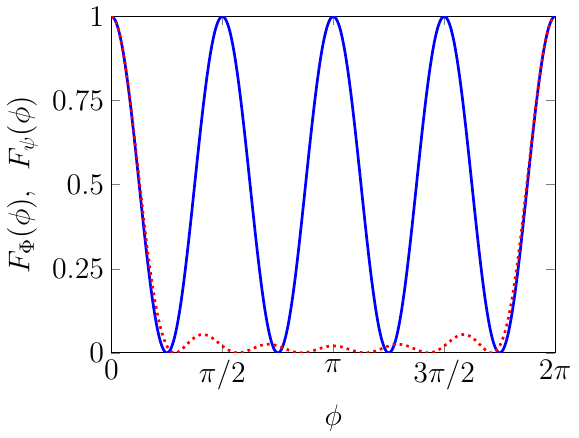}
    \caption{\label{fig:comparison}\textbf{Perfect quantum protractor versus spin 2-anticoherent state.} The overlap $F_\Phi(\phi)$ between spin 2-anticoherent state $\ket{\Phi}$ for spin-3 system and its version rotated around any axis $\hat{\v{n}}$ by angle $\phi$ (blue solid line) compared with the overlap $F_\psi(\phi)$ between perfect quantum protractor $\ket{\psi}$ for spin-3 system and its version rotated around an axis $k$ by angle $\phi$ with $k\in\{x,y,z\}$ (red dotted line).}
\end{figure}

One of the examples of the spin 2-anticoherent state of a spin-3 system is given by~\cite{goldberg2018quantum}:
\begin{equation}
    \ket{\Phi}=\frac{1}{\sqrt{2}}(\ket{3,2}_z-\ket{3,-2}_z).
\end{equation}
Its overlap with a version of itself rotated around an arbitrary axis $\hat{\v{n}}$ by an angle $\phi$ is given by
\begin{equation}
    F_{\Phi}(\phi):=|\langle \Phi | R_{\hat{\v{n}}}(\phi)|\Phi \rangle|^2 = \cos^2(2\phi),
\end{equation}
and is plotted in Fig.~\ref{fig:comparison}. While the overlap quickly drops from one to zero (signalling the sensitivity of $\ket{\Phi}$ to small angle rotations), it then oscillates just between two orthogonal states. As a result, it can only be used to distinguish between $\phi\in[0,\pi/4]$. For example, a rotation by any of the angles $\{\pi/4,3\pi/4,5\pi/4,7\pi/4\}$ gives not only the same overlap with the original state, but actually exactly the same quantum state, so there does not exist a measurement that could distinguish between those rotation angles.

On the other hand, the overlap of a perfect quantum protractor of spin-3 system, given by Eq.~\eqref{eq:state3}, with a version of itself rotated around an axis $k\in\{x,y,z\}$ by an angle $\phi$ is given by\footnote{More generally, this overlap for a perfect quantum protractor of a spin-$j$ system can be derived using Eq.~\eqref{eq:optimal1}, and is given by 
\begin{equation}
    |\langle \psi | R_{k}(\phi)|\psi \rangle|^2=\left[\frac{1}{2j+1}\left(1+2\sum_{m=1}^j \cos(m\phi)\right)\right]^2.
\end{equation}
}
\begin{align}
     F_\psi(\phi)&:=|\langle \psi | R_k(\phi)|\psi \rangle|^2 \nonumber\\
     &= \frac{1}{49}(1+2(\cos(\phi)+\cos(2\phi)+\cos(3\phi)))^2.
\end{align}
As can be seen in Fig.~\ref{fig:comparison}, it is slightly less sensitive to small angle rotations than the spin 2-anticoherent state (its overlap drops a little slower with increasing~$\phi$). However, it then rotates through seven mutually orthogonal states. Thus, performing a measurement in this basis allows one to distinguish between all angles $\phi\in[0,2\pi)$. For example, a rotation by angles $\{2\pi/7,4\pi/7,6\pi/7,8\pi/7,10\pi/7,12\pi/7\}$ not only gives a zero overlap with the initial state $\ket{\psi}$, but actually produces six more mutually orthogonal states that can be unambiguously distinguished by a single-shot measurement. 

\bibliographystyle{quantum}
\bibliography{bibliography.bib}

\begin{thebibliography}{10}

\bibitem{heisenberg1927anschaulichen}
Werner Heisenberg.
\newblock ``{\"U}ber den anschaulichen {I}nhalt der quantentheoretischen
  {K}inematik und {M}echanik''.
\newblock \href{https://dx.doi.org/10.1007/BF01397280}{Z. Phys. {\bf 43},
  172--198}~(1927).

\bibitem{robertson1929uncertainty}
Howard~Percy Robertson.
\newblock ``The uncertainty principle''.
\newblock \href{https://dx.doi.org/10.1103/PhysRev.34.163}{Phys. Rev. {\bf 34},
  163}~(1929).

\bibitem{busch2014colloquium}
Paul Busch, Pekka Lahti, and Reinhard~F Werner.
\newblock ``Colloquium: Quantum root-mean-square error and measurement
  uncertainty relations''.
\newblock \href{https://dx.doi.org/10.1103/RevModPhys.86.1261}{Rev. Mod. Phys.
  {\bf 86}, 1261}~(2014).

\bibitem{coles2017entropic}
Patrick~J Coles, Mario Berta, Marco Tomamichel, and Stephanie Wehner.
\newblock ``Entropic uncertainty relations and their applications''.
\newblock \href{https://dx.doi.org/10.1103/RevModPhys.89.015002}{Rev. Mod.
  Phys. {\bf 89}, 015002}~(2017).

\bibitem{mandelstam1991uncertainty}
L.~Mandelstam and I.~G. Tamm.
\newblock ``The uncertainty relation between energy and time in
  non-relativistic quantum mechanics''.
\newblock In Selected Papers.
\newblock Pages 115--123.
\newblock Springer~(1991).

\bibitem{peres1980measurement}
Asher Peres.
\newblock ``Measurement of time by quantum clocks''.
\newblock \href{https://dx.doi.org/10.1119/1.12061}{Am. J. Phys. {\bf 48},
  552--557}~(1980).

\bibitem{buzek1999optimal}
Vladimir Bu{\vv{z}}ek, Radoslav Derka, and Serge Massar.
\newblock ``Optimal quantum clocks''.
\newblock \href{https://dx.doi.org/10.1103/PhysRevLett.82.2207}{Phys. Rev.
  Lett. {\bf 82}, 2207}~(1999).

\bibitem{barnett1992limiting}
Stephen~M Barnett and DT~Pegg.
\newblock ``Limiting procedures for the optical phase operator''.
\newblock \href{https://dx.doi.org/10.1080/09500349214552141}{J. Mod. Opt. {\bf
  39}, 2121--2129}~(1992).

\bibitem{korzekwa2019distinguishing}
Kamil Korzekwa, Stanis{\l}aw Czach{\'o}rski, Zbigniew Pucha{\l}a, and Karol
  {\.Z}yczkowski.
\newblock ``Distinguishing classically indistinguishable states and channels''.
\newblock \href{https://dx.doi.org/10.1088/1751-8121/ab30f7}{J. Phys. A {\bf
  52}, 475303}~(2019).

\bibitem{korzekwa2014operational}
Kamil Korzekwa, David Jennings, and Terry Rudolph.
\newblock ``Operational constraints on state-dependent formulations of quantum
  error-disturbance trade-off relations''.
\newblock \href{https://dx.doi.org/10.1103/PhysRevA.89.052108}{Phys. Rev. A
  {\bf 89}, 052108}~(2014).

\bibitem{idel2015sinkhorn}
Martin Idel and Michael~M Wolf.
\newblock ``Sinkhorn normal form for unitary matrices''.
\newblock \href{https://dx.doi.org/10.1016/j.laa.2014.12.031}{Linear Algebra
  Appl. {\bf 471}, 76--84}~(2015).

\bibitem{dammeier2015uncertainty}
Lars Dammeier, Ren{\'e} Schwonnek, and Reinhard~F Werner.
\newblock ``Uncertainty relations for angular momentum''.
\newblock \href{https://dx.doi.org/10.1088/1367-2630/17/9/093046}{New J. Phys.
  {\bf 17}, 093046}~(2015).

\bibitem{radcliffe1971some}
J~Michael Radcliffe.
\newblock ``Some properties of coherent spin states''.
\newblock \href{https://dx.doi.org/10.1088/0305-4470/4/3/009}{J. Phys. A {\bf
  4}, 313}~(1971).

\bibitem{aragone1976intelligent}
C~Aragone, E~Chalbaud, and S~Salamo.
\newblock ``On intelligent spin states''.
\newblock \href{https://dx.doi.org/10.1063/1.522835}{J. Math. Phys. {\bf 17},
  1963--1971}~(1976).

\bibitem{renyi1961measures}
Alfr\'ed R\'enyi.
\newblock ``On measures of entropy and information''.
\newblock In Berkeley Symp. on Math. Statist. and Prob.
\newblock Pages 547--561.
\newblock ~(1961).
\newblock
  url:~\url{https://static.renyi.hu/renyi_cikkek/1961_on_measures_of_entropy_and_information.pdf}.

\bibitem{maassen1988generalized}
Hans Maassen and Jos~BM Uffink.
\newblock ``Generalized entropic uncertainty relations''.
\newblock \href{https://dx.doi.org/10.1103/PhysRevLett.60.1103}{Phys. Rev.
  Lett. {\bf 60}, 1103--1106}~(1988).

\bibitem{sanchez1993entropic}
Jorge S{\'a}nchez.
\newblock ``Entropic uncertainty and certainty relations for complementary
  observables''.
\newblock \href{https://dx.doi.org/10.1016/0375-9601(93)90269-6}{Phys. Lett. A
  {\bf 173}, 233--239}~(1993).

\bibitem{sanchez1995improved}
Jorge S{\'a}nchez-Ruiz.
\newblock ``Improved bounds in the entropic uncertainty and certainty relations
  for complementary observables''.
\newblock \href{https://dx.doi.org/10.1016/0375-9601(95)00219-S}{Phys. Lett. A
  {\bf 201}, 125--131}~(1995).

\bibitem{puchala2015certainty}
Zbigniew Pucha{\l}a, {\L}ukasz Rudnicki, Krzysztof Chabuda, Miko{\l}aj
  Paraniak, and Karol {\.Z}yczkowski.
\newblock ``Certainty relations, mutual entanglement, and nondisplaceable
  manifolds''.
\newblock \href{https://dx.doi.org/10.1103/PhysRevA.92.032109}{Phys. Rev. A
  {\bf 92}, 032109}~(2015).

\bibitem{zimba2006anticoherent}
Jason Zimba.
\newblock ``Anticoherent spin states via the {M}ajorana representation''.
\newblock Electron. J. Theor. Phys. {\bf 3}, 143--156~(2006).
\newblock
  url:~\url{https://citeseerx.ist.psu.edu/document?repid=rep1&type=pdf&doi=2b4f1377d81f87e33c3dc62d2ebcf83cae2b0aee}.

\bibitem{goldberg2018quantum}
Aaron~Z Goldberg and Daniel~FV James.
\newblock ``Quantum-limited {E}uler angle measurements using anticoherent
  states''.
\newblock \href{https://dx.doi.org/10.1103/PhysRevA.98.032113}{Phys. Rev. A
  {\bf 98}, 032113}~(2018).

\bibitem{scott2004multipartite}
Andrew~J Scott.
\newblock ``Multipartite entanglement, quantum-error-correcting codes, and
  entangling power of quantum evolutions''.
\newblock \href{https://dx.doi.org/10.1103/PhysRevA.69.052330}{Phys. Rev. A
  {\bf 69}, 052330}~(2004).

\bibitem{giovannetti2011advances}
Vittorio Giovannetti, Seth Lloyd, and Lorenzo Maccone.
\newblock ``Advances in quantum metrology''.
\newblock \href{https://dx.doi.org/10.1038/nphoton.2011.35}{Nat. Photon. {\bf
  5}, 222--229}~(2011).

\bibitem{demkowicz2011optimal}
Rafa{\l} Demkowicz-Dobrza{\'n}ski.
\newblock ``Optimal phase estimation with arbitrary a priori knowledge''.
\newblock \href{https://dx.doi.org/10.1103/PhysRevA.83.061802}{Phys. Rev. A
  {\bf 83}, 061802}~(2011).

\bibitem{hall2012heisenberg}
Michael~JW Hall and Howard~M Wiseman.
\newblock ``Heisenberg-style bounds for arbitrary estimates of shift parameters
  including prior information''.
\newblock \href{https://dx.doi.org/10.1088/1367-2630/14/3/033040}{New J. Phys.
  {\bf 14}, 033040}~(2012).

\bibitem{kolenderski2008optimal}
Piotr Kolenderski and Rafal Demkowicz-Dobrzanski.
\newblock ``Optimal state for keeping reference frames aligned and the
  {P}latonic solids''.
\newblock \href{https://dx.doi.org/10.1103/PhysRevA.78.052333}{Phys. Rev. A
  {\bf 78}, 052333}~(2008).

\bibitem{albarelli2022probe}
Francesco Albarelli and Rafa{\l} Demkowicz-Dobrza{\'n}ski.
\newblock ``Probe incompatibility in multiparameter noisy quantum metrology''.
\newblock \href{https://dx.doi.org/10.1103/PhysRevX.12.011039}{Phys. Rev. X
  {\bf 12}, 011039}~(2022).

\bibitem{helstrom1969quantum}
Carl~W Helstrom.
\newblock ``Quantum detection and estimation theory''.
\newblock \href{https://dx.doi.org/10.1007/BF01007479}{J. Stat. Phys. {\bf 1},
  231--252}~(1969).

\bibitem{rzkadkowski2017discrete}
Wojciech Rz{\k{a}}dkowski and Rafa{\l} Demkowicz-Dobrza{\'n}ski.
\newblock ``Discrete-to-continuous transition in quantum phase estimation''.
\newblock \href{https://dx.doi.org/10.1103/PhysRevA.96.032319}{Phys. Rev. A
  {\bf 96}, 032319}~(2017).

\bibitem{Auzinsh2010OpticallyInteractions}
M.~Auzinsh, D.~Budker, and S.~Rochester.
\newblock ``{Optically Polarized Atoms: Understanding Light-atom
  Interactions}''.
\newblock Oxford University Press. ~(2010).

\bibitem{Kopciuch2024optimized}
Marek Kopciuch, Magdalena Smolis, Adam Miranowicz, and Szymon Pustelny.
\newblock ``Optimized optical tomography of quantum states of a
  room-temperature alkali-metal vapor''.
\newblock \href{https://dx.doi.org/10.1103/PhysRevA.109.032402}{Phys. Rev. A
  {\bf 109}, 032402}~(2024).

\bibitem{Kopciuch2022}
Marek Kopciuch and Szymon Pustelny.
\newblock ``Optical reconstruction of the collective density matrix of a
  qutrit''.
\newblock \href{https://dx.doi.org/10.1103/PhysRevA.106.022406}{Phys. Rev. A
  {\bf 106}, 022406}~(2022).

\bibitem{chryssomalakos2017optimal}
C~Chryssomalakos and H~Hern{\'a}ndez-Coronado.
\newblock ``Optimal quantum rotosensors''.
\newblock \href{https://dx.doi.org/10.1103/PhysRevA.95.052125}{Phys. Rev. A
  {\bf 95}, 052125}~(2017).

\bibitem{martin2020optimal}
John Martin, Stefan Weigert, and Olivier Giraud.
\newblock ``Optimal detection of rotations about unknown axes by coherent and
  anticoherent states''.
\newblock \href{https://dx.doi.org/10.22331/q-2020-06-22-285}{Quantum {\bf 4},
  285}~(2020).

\bibitem{bouchard2017quantum}
Frederic Bouchard, P~de~la Hoz, Gunnar Bj{\"o}rk, RW~Boyd, Markus Grassl,
  Z~Hradil, E~Karimi, AB~Klimov, Gerd Leuchs, J~{\vv{R}}eh{\'a}{\vv{c}}ek,
  et~al.
\newblock ``Quantum metrology at the limit with extremal majorana
  constellations''.
\newblock \href{https://dx.doi.org/10.1364/OPTICA.4.001429}{Optica {\bf 4},
  1429--1432}~(2017).

\bibitem{ferretti2024generating}
Hugo Ferretti, Y~Batuhan Yilmaz, Kent Bonsma-Fisher, Aaron~Z Goldberg, Noah
  Lupu-Gladstein, Arthur~OT Pang, Lee~A Rozema, and Aephraim~M Steinberg.
\newblock ``Generating a 4-photon tetrahedron state: toward simultaneous
  super-sensitivity to non-commuting rotations''.
\newblock \href{https://dx.doi.org/10.1364/OPTICAQ.510125}{Opt. Quantum {\bf
  2}, 91--102}~(2024).

\bibitem{zauner2011quantum}
Gerhard Zauner.
\newblock ``Quantum designs: Foundations of a noncommutative design theory''.
\newblock \href{https://dx.doi.org/10.1142/S0219749911006776}{Int. J. Quantum
  Inf. {\bf 9}, 445--507}~(2011).

\bibitem{renes2004symmetric}
J.~M. Renes, R.~Blume-Kohout, A.~J. Scott, and C.~M. Caves.
\newblock ``Symmetric informationally complete quantum measurements''.
\newblock \href{https://dx.doi.org/10.1063/1.1737053}{J. Math. Phys. {\bf 45},
  2171--2180}~(2004).

\bibitem{rudzinski2024orthonormal}
Marcin Rudzi{\'n}ski, Adam Burchardt, and Karol {\.Z}yczkowski.
\newblock ``Orthonormal bases of extreme quantumness''.
\newblock \href{https://dx.doi.org/10.22331/q-2024-01-25-1234}{Quantum {\bf 8},
  1234}~(2024).

\bibitem{wootters1989optimal}
William~K Wootters and Brian~D Fields.
\newblock ``Optimal state-determination by mutually unbiased measurements''.
\newblock \href{https://dx.doi.org/10.1016/0003-4916(89)90322-9}{Ann. Phys.
  {\bf 191}, 363--381}~(1989).

\bibitem{giovannetti2004quantum}
Vittorio Giovannetti, Seth Lloyd, and Lorenzo Maccone.
\newblock ``Quantum-enhanced measurements: beating the standard quantum
  limit''.
\newblock \href{https://dx.doi.org/10.1126/science.1104149}{Science {\bf 306},
  1330--1336}~(2004).

\end{thebibliography}

\end{document}